%% file: Main.tex
\documentclass[a4paper,UKenglish,cleveref, autoref]{lipics-v2019}

\title{Graph Balancing with Orientation Costs} 

\titlerunning{Graph Balancing with Orientation Costs}

\author{Roy Schwartz}{Technion -- Israel Institute of Technology, Israel}{schwartz@cs.technion.ac.il}{}{}

\author{Ran Yeheskel}{Technion -- Israel Institute of Technology, Israel}{ran.yeheskel11@gmail.com}{}{}

\authorrunning{R. Schwartz and R. Yeheskel}

\Copyright{Roy Schwartz and Ran Yeheskel}



\nolinenumbers 

\hideLIPIcs  

\EventEditors{Michael A. Bender, Ola Svensson, and Grzegorz Herman}
\EventNoEds{3}
\EventLongTitle{27th Annual European Symposium on Algorithms (ESA 2019)}
\EventShortTitle{ESA 2019}
\EventAcronym{ESA}
\EventYear{2019}
\EventDate{September 9--11, 2019}
\EventLocation{Munich/Garching, Germany}
\EventLogo{}
\SeriesVolume{144}
\ArticleNo{78}

\input{our_macros.tex}

\usepackage{xspace}
\usepackage{amsmath}
\usepackage{amsthm}
\usepackage{amssymb}
\usepackage{nicefrac}
\usepackage{tikz}
\usepackage{pgfplots}
\usepackage{graphicx}
\usepackage{algorithmic}
\usepackage[ruled,vlined,linesnumbered]{algorithm2e}
\newtheorem{observation}{Observation}

\begin{document}

\maketitle

\begin{abstract}
Motivated by the classic \GAP, we consider the \GBnp problem in the presence of orientation costs: given an undirected multi-graph $G=(V,E)$ equipped with edge weights and orientation costs on the edges, the goal is to find an orientation of the edges that minimizes both the maximum weight of edges oriented toward any vertex (makespan) and total orientation cost.
We present a general framework for minimizing makespan in the presence of costs that allows us to:
$(1)$ achieve bicriteria approximations for the \GBnp problem that capture known previous results (Shmoys-Tardos [Math. Progrm. `93], Ebenlendr-Krc{\'{a}}l-Sgall [Algorithmica `14], and Wang-Sitters [Inf. Process. Lett. `16]); and $(2)$ achieve bicriteria approximations for extensions of the \GBnp problem that admit hyperedges and unrelated weights.
Our framework is based on a remarkably simple rounding of a strengthened linear relaxation.
We complement the above by presenting bicriteria lower bounds with respect to the linear programming relaxations we use that show that a loss in the total orientation cost is required if one aims for an approximation better than $2$ in the makespan.
\end{abstract}

\section{Introduction}\label{sec:Introduction}
\input{Introduction.tex}

\noindent {\bf Paper Organization:} Section \ref{sec:Preliminaries} contains the required preliminaries.
In Section \ref{sec:GraphBalancing} we present our general framework and apply it to \sGB to obtain bicriteria algorithms.
Section \ref{sec:GBLowerBound} contains our bicriteria lower bound for \sGB.
Finally, in Section \ref{Section:Extensions} we consider the mentioned extensions of \sGB and apply the framework to these extensions to obtain improved algorithms.

\section{Preliminaries}\label{sec:Preliminaries}
\input{Preliminaries.tex}

\section{The General Framework and Graph Balancing}\label{sec:GraphBalancing}
\input{Algorithm.tex}
\subsection{Graph Balancing -- Upper Bound on Tradeoff Between Makespan and Orientation Cost}
\input{Analysis.tex}
\subsection{Graph Balancing -- Extending the Tradeoff}\label{Section:ImprovingUpperBound}
\input{ImprovedGB.tex}

\section{A Simple Algorithm for Graph Balancing with Orientation Costs}
\input{NaiveGB.tex}

\section{Lower Bound on The Tradeoff Between Makespan and Cost}\label{sec:GBLowerBound}
\input{IntegralityGapCost.tex}

\section{Extending Graph Balancing to Hyperedges and Unrelated Weights} \label{Section:Extensions}
\subsection{Graph Balancing with Unrelated Light Hyperedges} \label{Section:GBH}
\input{GBUnrelatedLightHyper.tex}

\subsection{Graph Balancing with Unrelated Light Hyperedges and Unrelated Heavy Edges}\label{sec:GBU}
\input{GBUnrelatedLightHyperUnrelatedHeavyEdge.tex}

\subsection{Semi-Related Graph Balancing} \label{Section:SRGB}
\input{SemiRestrictedGB.tex}

\bibliography{sources}

\end{document}

%% file: our_macros.tex
\newcommand{\LS}{\textit{Local Step}\xspace}
\newcommand{\GS}{\textit{Global Step}\xspace}

\newcommand{\LPgap}{\textsc{$LP_{GAP}$}\xspace}
\newcommand{\LPthree}{\textsc{$LP$}\xspace}
\newcommand{\LPstong}{\textsc{$LP_k$}\xspace}

\newcommand{\sRA}{\textsc{Ra}\xspace}

\newcommand{\sSRGB}{\textsc{Srgb($c$)}\xspace}
\newcommand{\SRGB}{\textsc{Semi-Related Graph Balancing}\xspace}
\newcommand{\sGB}{\textsc{Gb}\xspace}
\newcommand{\sGAP}{\textsc{Gap}\xspace}

\newcommand{\sGBUH}[1]{\textsc{Gbuh(#1)}\xspace}
\newcommand{\sGBU}[1]{\textsc{Gbu(#1)}\xspace}
\newcommand{\GBnp}{\textsc{Graph Balancing}\xspace}

\newcommand{\SANTA}{\textsc{Santa Claus Problem}\xspace}
\newcommand{\GAP}{\textsc{Generalized Assignment Problem}\xspace}

\newcommand{\RA}{\textsc{Restricted Assignment} problem\xspace}

\newcommand{\GBULH}{\textsc{Graph Balancing with Unrelated Light Hyper Edges}\xspace}
\newcommand{\GBULHUH}{\textsc{Graph Balancing with Unrelated Light Hyper Edges and Unrelated Heavy Edges}\xspace}

\newcommand{\bx}{$\textbf{x}$\xspace}

\newcommand{\ie}{{\em i.e.}\xspace}

%% file: Introduction.tex
We consider the \GBnp problem (\sGB) where we are given an undirected multi-graph $G=(V,E)$ equipped with edge weights $p:E\rightarrow\mathbb{R}^{+}$. The goal is to orient all the edges of the graph, where each edge can be oriented to one of its endpoints. Given an orientation of the edges the load of a vertex $u$ is the sum of weights of edges oriented toward it. The goal is to find an orientation of the edges that minimizes the maximum load over all vertices.

\sGB was first introduced by Ebenlendr {\em et al.} \cite{EbenlendrKS08} and since its introduction it has attracted much attention (see, {\em e.g.}, \cite{JansenR18, WangS16, HuangO16, ChakrabartyS16, PageS16}). Besides being a natural graph optimization problem on its own, a main motivation for considering \sGB is the well known \GAP (\sGAP) (see, {\em e.g.}, \cite{ShmoysT93, EbenlendrKS08, Svensson11, WilliamsonSBook}). 
In \sGAP we are given a collection $\mathcal{M}$ of $m$ machines and a collection $\mathcal{J}$ of $n$ jobs, along with processing times $p_{i,j}$ (the processing time of job $j$ on machine $i$) and assignment costs $c_{i,j}$ (the cost of assigning job $j$ to machine $i$).
Each job must be assigned to one of the machines.
The processing time of machine $i$ is the sum of processing times $p_{i,j}$ over all jobs $j$ that are assigned to $i$, and the makespan of an assignment is the maximum over all machines $i$ of its processing time.
Additionally, the total assignment cost of an assignment is the sum of assignment costs $c_{i,j}$ over all machines $i$ and jobs $j$ that are assigned to $i$.
Given a target makespan $T$, we denote by $C(T)$ the minimum total assignment cost over all assignments with makespan at most $T$.
If there are no assignments with makespan at most $T$, then $C(T)=\infty$.
The goal in \sGAP, given a target makespan $T$, is to find an assignment with makespan at most $T$ and total assignment cost at most $C(T)$, or declare that no such assignment exists.
We note that only $T$ is given to the algorithm whereas $C(T)$ is not.
For this bicriteria problem, the celebrated result of Shmoys and Tardos \cite{ShmoysT93} provides an approximation algorithm that finds an assignment with makespan at most $2T$ and total assignment cost at most $C(T)$.

\sGB is a captured by \sGAP since one can: $(1)$ set $\mathcal{M}$ to be $V$ and $\mathcal{J}$ to be $E$; and $(2)$ for each job $j\in \mathcal{J}$ (which corresponds to an edge $e\in E$) set its processing time to be $p_e$ for the two machines that correspond to the endpoints of $e$ and $\infty$ for all other machines.
Note that assigning job $j$ to machine $i$ corresponds to orienting the edge $e$ toward its endpoint that corresponds to machine $i$.
There are two important things to note. First, \sGB was originally defined as a \textit{single criterion} optimization problem as opposed to \sGAP which is a bicriteria optimization problem. Second, the weights $p$ in \sGB, which represent the processing times of the jobs, are \textit{related}, \ie, the processing times do not depend on the vertex the edge is oriented to.
Ebenlendr {\em et al.} \cite{EbenlendrKS08} introduced a novel linear relaxation and rounding algorithm that achieves an approximation of 1.75 with respect to the optimal makespan. They also proved that even for this special case, no polynomial time algorithm can achieve an approximation less than $1.5$ unless $P=NP$, thus extending the hardness of \sGAP to \sGB.

In this work we consider the {\em bicriteria} \sGB problem, where we are also given orientation costs, the equivalent to the assignment costs in \sGAP. Formally, an edge $e=(u,v)$ has orientation costs $c_{e,u}$ and  $c_{e,v}$ and orienting it to $u$ incurs a cost of $c_{e,u}$. Similarly to \sGAP, given a target makespan $T$, the goal is to find an orientation of the edges with total orientation cost at most $C(T)$ and makespan at most $T$.\footnote{As in \sGAP, the total orientation cost of an orientation is defined as the sum of orientation costs $c_{e,u}$ over all vertices $u$ and edges $e$ oriented toward $u$. $C(T)$ is defined as the minimum total orientation cost over all orientations with makespan at most $T$. If no such orientation exists then $C(T)$ is set to $\infty$.} To the best of our knowledge, the bicriteria \sGB problem was not previously considered.
We say that an algorithm is a $(\alpha, \beta)$-approximation if given a target makespan $T$, it outputs an orientation with makespan at most $\alpha T$ and total orientation cost at most $\beta C(T)$. Thus, \cite{ShmoysT93} is a $(2,1)$-approximation to \sGB. 
We note that the algorithm of \cite{EbenlendrKS08} cannot handle orientation costs and is in fact a $(1.75,\infty)$-approximation for \sGB.
A result by Wang and Sitters \cite{WangS16} implicitly gives a $(\nicefrac{11}{6}, \nicefrac{3}{2})$-approximation for \sGB.


We study the bicriteria tradeoff between makespan and total orientation cost in \sGB, presenting both upper and lower bounds (the latter are with respect to the linear programming relaxations used in this work). We employ a remarkably simple general framework that allows us to achieve bicriteria approximations for \sGB that capture and extend known results. Furthermore, we consider extensions of \sGB that allow for: (1) hyperedges to be present, \ie, a job can be assigned to more than two machines; and (2) processing times can be unrelated, \ie, the processing time of a job might depend on the machine it is assigned to. Our results regarding these extensions improve upon the previously best known algorithms, and are also based on the general framework presented in this paper.
We believe this framework might be of independent interest to other related scheduling problems.

\subsection{Our Results}
Our results are of three different flavors: bicriteria upper bounds for \sGB, bicriteria lower bounds for \sGB, and both upper and lower bicriteria bounds for extensions of \sGB (all lower bounds are with respect to the linear programming relaxations we use). Let us now elaborate on each of the above.

\vspace{3pt}
\noindent \textbf{Upper Bounds:} We present a general framework for minimizing makespan in the presence of costs and obtain two algorithms that achieve bicriteria approximations for \sGB. This is summarized in the following two theorems.
\begin{theorem} \label{Theorem:firstUpperBound}
There exists a polynomial time algorithm that finds an orientation that is a $(1.75+\gamma, \nicefrac{1}{(2\gamma+0.5)})$-approximation for \GBnp, for every $\nicefrac{1}{12}-\epsilon\leq \gamma \leq \nicefrac{1}{4}$ where $\epsilon=\nicefrac[]{\sqrt{33}}{4}-\nicefrac[]{17}{12}\approx 0.103$. 
\end{theorem}
\begin{theorem} \label{Theorem:secondUpperBound}
There exists a polynomial time algorithm that finds an orientation that is a $(1.75+\gamma, 1+\nicefrac{1}{\gamma})$-approximation for \GBnp, for every $0\leq \gamma \leq \nicefrac{1}{4}$.
\end{theorem}
Both the above theorems provide a smooth tradeoff between makespan and orientation cost while capturing previous known results for \sGB as special cases, {\em i.e.}, Theorem \ref{Theorem:firstUpperBound} captures the $(2,1)$ and $(\nicefrac{11}{6}, \nicefrac{3}{2})$ approximations of \cite{ShmoysT93} and \cite{WangS16} for $\gamma = \nicefrac[]{1}{4}$ and $\gamma = \nicefrac[]{1}{12}$ respectively, whereas Theorem \ref{Theorem:secondUpperBound} captures the $(1.75,\infty)$-approximation of \cite{EbenlendrKS08} for $\gamma =0$.
Theorem \ref{Theorem:firstUpperBound} is depicted in Figure \ref{OurResultsGB}.

\vspace{3pt}
\noindent \textbf{Lower Bounds:} We present bicirteria lower bounds for \sGB. As previously mentioned, our lower bounds apply to a strengthening of the relaxation of \cite{EbenlendrKS08}, which we denote by \LPstong (see subsection \ref{Section:ImprovingUpperBound}). The lower bound is summarized in the follwing theorem and is depicted in Figure \ref{OurResultsGB}.
\begin{theorem} \label{Theorem:LowerBound}
For every $0\leq \gamma < \nicefrac[]{1}{4}$ and $\epsilon>0$, there exists an instance for \GBnp and target makespan $T$ such that: (1) \LPstong is feasible and has value of $OPT_{\LPstong}$, and (2) every orientation whose makespan is at most $(1.75+\gamma)T$ has orientation cost of at least $\nicefrac{1}{(\gamma+0.75+\epsilon)} OPT_{\LPstong}$.
\end{theorem}
To the best of our knowledge, all algorithms for \sGB that find an orientation that achieves an approximation better than $2$ with respect to the makespan use the relaxation of \cite{EbenlendrKS08} (or no relaxation at all, {\em e.g.}, \cite{HuangO16}). \footnote{Recently, Jansen and Rohwedder \cite{JansenR18} showed that using a different stronger relaxation called the configuration LP one can achieve an approximation of less than $1.75$ to the makespan. However, this result does not produce a polynomial time algorithm that orients the edges but rather only approximates the {\em value} of the optimal makespan.
Moreover, this result has an unbounded loss with respect to the orientation cost.}

\vspace{3pt}
\noindent \textbf{Extensions:} Using our general framework, we present bicriteria algorithms for extensions of \sGB. The extensions of \sGB we consider allow hyperedges and unrelated weights to the edges. It is important to note that all the upper bounds presented below hold for the single criterion versions of these problems as well.
In particular, we achieve an approximation strictly better than 2, with respect to the makespan, to several problems that capture \sGB and are not captured by the \RA (\sRA).\footnote{The \sRA is a special case of \sGAP where each job has a set of machines it can be assigned to, and has an equal processing time on each of them.} To the best of our knowledge, this is the first polynomial time algorithm with approximation factor better than 2 to the makespan for problems that capture \sGB and are not captured by \sRA. Let us now elaborate on these extensions.

The first extension allows for light unrelated hyperedges. Formally, given $\beta\in [0,1]$, the input can contain hyperedges whose weight with respect to the vertices it shares may vary, as long as it does not exceed $\beta$ (we may assume without loss of generality that the largest weight in $p$ equals $1$). 
We denote this problem by \GBULH (\sGBUH{$\beta$}). 
A special case of this problem was introduced by Huang and Ott in \cite{HuangO16} who presented a $(\nicefrac[]{5}{3}+\nicefrac[]{\beta}{3},\infty)$-approximation when $\beta \in [\nicefrac{4}{7}, 1)$.
We improve upon \cite{HuangO16} in three aspects.
First, we consider the general bicriteria problem, {\em i.e.}, orientation costs are present, and achieve bounded loss with respect to the total orientation cost (recall that \cite{HuangO16} cannot handle orientation costs).
Second, we allow any $\beta \in [0,1]$, where \cite{HuangO16} allows for $\beta \in [\nicefrac{4}{7}, 1)$ only.
Third, we allow the hyperedges to be unrelated, {\em i.e.}, different weights to different endpoints, where hyperedge weights in \cite{HuangO16} are related.
Our result for this extension is summarized in the following theorem.
\begin{theorem} \label{Theorem:GBUH}
Let $ 0\leq \beta \leq 1$. For every $max\left\{ \nicefrac{1}{12}, \nicefrac{\beta}{3}-\nicefrac{1}{12} \right\} \leq \gamma\leq \nicefrac{1}{4}$, there exists a polynomial time algorithm that finds an orientation that is a $(1.75+\gamma, \nicefrac{1}{(2\gamma+0.5)})$-approximation to \sGBUH{$\beta$}.

\end{theorem}

The second extension further generalizes the first one, and it also allows edges to have unrelated weights as long as the weights are greater than $\beta$. Unfortunately, we prove that this problem in its full generality is as hard to approximate as \sGAP. However, if it is assumed that the optimal makespan is at least $1$ (as before we can assume without loss of generality that the largest weight in $p$ equals $1$), we can achieve improved results.
We denote this problem by \GBULHUH (\sGBU{$\beta$}).
\footnote{While the assumption that the largest weight $p$ equals $1$ implies that the optimal makespan is least $1$ for \sGB and \sGBUH{$\beta$}, this is not necessarily the case when the edge weights might be unrelated.
Thus, the assumption in \sGBU{$\beta$} that the optimal makespan is at least $1$ is not without loss of generality.}
\begin{theorem} \label{Theorem:GBU}
Let $\beta\geq \sqrt{2}-1$. For every $\nicefrac{\beta}{3}-\nicefrac{1}{12} \leq \gamma\leq \nicefrac{1}{4}$, there exists a polynomial time algorithm that finds an orientation that is a $(1.75+\gamma,\nicefrac{1}{(2\gamma+0.5)})$-approximation to \sGBU{$\beta$}.
\end{theorem}
We prove that there are values of $\beta$ for which the bicriteria approximation of Theorem \ref{Theorem:GBU} is tight.
Specifically, we prove the latter for $\beta=\nicefrac[]{1}{2}$ and \LPstong.
The lower bounds are summarized in the following theorem.
\begin{theorem} \label{Theorem:GBUintegrality}
For every $\epsilon>0$, there exists an instance of \sGBU{$0.5$} that is feasible to \LPstong and every orientation has a makespan of at least $\nicefrac{11}{6}-\epsilon$. Moreover, for every $\nicefrac{1}{12}\leq \gamma\leq \nicefrac{1}{4}$, target makespan $T$ and $\epsilon>0$, there exists an instance for \sGBU{$0.5$} that is feasible to \LPstong and has a value of $OPT_{\LPstong}$, and every orientation with makespan at most $(1.75+\gamma)T$ has an orientation cost of at least $\nicefrac{(1-\epsilon)}{(2\gamma+0.5)}\cdot OPT_{\LPstong}$.
\end{theorem}

In the third and final extension we allow the edges in \sGB to be unrelated, but the weights cannot vary arbitrarily. Formally, given a parameter $c\geq 1$, every edge $e=(u,v)$ satisfies $p_{e,u}\leq c\cdot p_{e,v}$ and $p_{e,v}\leq c\cdot p_{e,u}$.
We denote this problem by \SRGB (\sSRGB).
The following theorem summarizes our algorithm for \sSRGB.
\begin{theorem} \label{Theorem:SRGB}
There exists a polynomial time algorithm to \sSRGB, that finds an orientation that is a $(1.5+0.5a, \nicefrac{1}{a})$-approximation, where $a$ is the root in the range $[0.5, 1]$ of the polynomial:
\[
\left(\nicefrac{1}{c}+\nicefrac{1}{2}\right)\cdot a^3+\left(\nicefrac{5}{(2c)}-\nicefrac{1}{2}\right)\cdot a^2 - \nicefrac{7}{(2c)}\cdot a + \nicefrac{1}{c}.
\]
\end{theorem}
We remark that the approximation guaranteed by Theorem \ref{Theorem:SRGB} is never worse than $2$ since it can be proved that $a=1-\Omega \left( \nicefrac[]{1}{c}\right)$, yielding a $(2-\Omega(\nicefrac[]{1}{c}),1+O(\nicefrac[]{1}{c}))$-approximation. 
It is worth noting that when $c=\infty$, which corresponds to the most general case, even the configuration LP has an integrality gap of $2$ with respect to the makespan (see \cite{EbenlendrKS14, VerschaeW14}).

\begin{center}
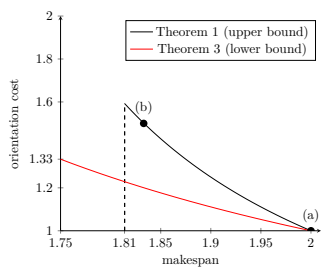
\begin{figure}
\begin{center}
	\begin{tikzpicture}[scale=0.5]
	\begin{axis}[
	axis lines = left,
	xlabel = {makespan},
	ylabel = {orientation cost},
	xmin=1.75, xmax=2.01,
	ymin=1, ymax=2,
	xtick={1.75, 1.81386, 1.85, 1.9, 1.95, 2},
	ytick={1, 1.2, 1.3333, 1.6, 1.8, 2}
	]

	
	\addplot [
	domain=1.81386:2,
	samples=100,
	color=black,
	]
	{1/(2*(x-1.75)+0.5)};
	\addlegendentry{Theorem \ref{Theorem:firstUpperBound} (upper bound)}
	
	\draw [dashed] (64, 0) -- (64, 59);
	
	\draw [dashed] (64, 0) -- (64, 59);
	
	\addplot [
	domain=1.75:2,
	samples=100,
	color=red,
	]
	{1/(x-1)};
	\addlegendentry{Theorem \ref{Theorem:LowerBound} (lower bound)}

	\node[label={(a)},circle,fill,inner sep=2pt] at (axis cs:2,1) {};
	\node[label={(b)},circle,fill,inner sep=2pt] at (axis cs:1.833, 1.5) {};
	
	\end{axis}
	\end{tikzpicture}
	\caption{Our bicriteria bounds for \GBnp. (a) is given in Shmoys and Tardos \cite{ShmoysT93}, whereas (b) is implicitly given in Wang and Sitters \cite{WangS16}.}
	\label{OurResultsGB}
\end{center}
\end{figure}

\end{center}

\subsection{Our Techniques}

We present a remarkably simple framework that allows us to provide bicriteria upper bounds for both \sGB and its extensions, \ie, \sGBUH{$\beta$}, \sGBU{$\beta$}, and \sSRGB.
The framework is based on rounding of a strengthening of the linear relaxation of \cite{EbenlendrKS08}.

The rounding is comprised of two complementary steps, the first {\em local} and the second {\em global}.
Intuitively, in the first local step, each edge can be oriented to one of its endpoints in case the relaxation indicates a strong (fractional) inclination toward that endpoint.
We note that in order to quantify this inclination the weight of the edge is taken into account, where lighter edges are less likely to be oriented.
Specifically, denote by $x_{e,u}\in [0,1]$ how much the relaxation fractionally orients edge $e=(u,v)$ toward its endpoint $u$.
The local step orients $e$ toward $u$ if $x_{e,u}> f(p_e)$ for some non-increasing threshold function $f:[0,1]\rightarrow [\nicefrac[]{1}{2},1]$.
As previously mentioned, this step is considered {\em local} since only $x_{e,u}$ and $p_e$ are used to determine whether to orient $e$, and if so to which of its two endpoints.\footnote{This rounding was used in Wang and Sitters \cite{WangS16} with a specific ``step'' threshold function $f$ to implicitly obtain a $(\nicefrac{11}{6}, \nicefrac{3}{2})$-approximation for \sGB.}
In the second global step of the rounding, we consider the remaining edges which were not yet oriented in the first local step and apply the algorithm of Shmoys and Tardos \cite{ShmoysT93} which finds a minimum cost perfect matching in a suitable bipartite graph.
As previously mentioned, this step is considered {\em global} since all edges which are not yet oriented are taken into consideration when computing the matching.

The above two-phase rounding is not sufficient on its own to obtain our claimed results, and we further strengthen the relaxation of \cite{EbenlendrKS08} by forcing additional new constraints.
Intuitively, for every vertex $u$ our constraints state that if a collection of edges $S$ touching $u$ has total weight of more than $T$ then not all edges in $S$ can be chosen.
We enforce the above constraints for all subsets of size at most $k$, for some fixed parameter $k$, resulting in a strengthened linear relaxation which we denote by $LP_k$.
It is important to note that these constraints cannot be inferred from the original relaxation of \cite{EbenlendrKS08}, and thus are required in our analysis of the above two-phase rounding.

%
%
%
%
%

\subsection{Additional Related Work}
Lenstra {\em et al.} \cite{LenstraST90} introduced the classic well known $2$-approximation to the single criterion \sGAP. They also proved that no polynomial time algorithm can approximate the makespan within a factor less than $1.5$ unless $P=NP$. This was followed by Shmoys and Tardos \cite{ShmoysT93} who introduced the bicriteria \sGAP and presented a $(2,1)$-approximation for it.
A slightly improved approximation of $2-\nicefrac[]{1}{m}$ for the makespan was given by Shchepin and Vakhania \cite{SV05}.
If the number of machines is fixed polynomial time approximation schemes are known \cite{HS76,JP01}.
For the case of uniformly related machines (each machine $i$ has speed $s_i$ and assigning job $j$ to machine $i$ takes $\nicefrac[]{p_j}{s_i}$ time) Hochbaum and Shmoys \cite{HS98} presented a polynomial time approximation scheme.
The \RA (\sRA) is a special case were each job has an equal processing time on the machines it can be assigned to (for every job $j$ and machine $i$: $p_{i,j}\in \{p_j,\infty\}$). For this special case, Svensson \cite{Svensson11} proved that one can approximate the value of the optimal makespan by a
factor of $\nicefrac{33}{17}$ using the configuration LP, that was first introduced by Bansal {\em et al.} for the \SANTA \cite{BansalS06}. 
This was subsequently improved by Jansen and Rohwedder \cite{JansenR17} who presented an approximation of $\nicefrac{11}{6}$.
If one further assumes that the processing times have only two possible values \cite{JansenLM18} presented an improved approximation of $\nicefrac{5}{3}$.
The above results \cite{Svensson11, JansenR17, JansenLM18} do not present polynomial time algorithms that produce a schedule with the promised makespan, but only approximate the value of the makespan.

When considering \sGB, Jansen and Rohwedder \cite{JansenR18} recently showed a similar flavor result: using the configuration LP one can estimate the value of the optimal makespan by a factor of $1.75-\epsilon$, for some small constant $\epsilon>0$. However, as before, \cite{JansenR18} does not produce an orientation in polynomial time.
The special case of \sGB where only two processing times are present admits a (tight) $1.5$-approximation (given independently by \cite{HuangO16, ChakrabartyS16, PageS16}).

To the best of our knowledge, no work on \sGB considered orientation costs and in particular the tradeoff between makespan and orientation cost.

%% file: Preliminaries.tex
Given a multi-graph $G=(V,E)$ and a vertex $u\in V$ denote by $\delta(u)\triangleq \{e\in E\mid u\in e\}$ the collection of edges incident to $u$. In addition define: $\mathcal{F}(u)\triangleq \{S\subseteq \delta(u) \mid \sum_{e\in S} p_e \leq 1\}$, {\em i.e.}, the collection of feasible subsets of edges incident to $u$ (for simplicity of presentation we further assume without loss of generality that $T=1$ since we can scale all processing times by $T$).
Moreover, we denote by $OPT_{LP}$ and $OPT_{LP_k}$ the optimal value of a feasible solution to the relaxation $LP$ and $LP_k$ respectively. 



The algorithm of Shmoys and Tardos \cite{ShmoysT93} is a key ingredient in our framework, thus we present it not only for completion but also since understanding its inner-working helps in analyzing our algorithms. Recall that \cite{ShmoysT93} is a $(2,1)$-approximation for \sGAP. We assume without loss of generality that $T=1$ since one can scale the processing times by $T$. First, the relaxation in Figure \ref{RelaxationST} is solved, where $\mathcal{J}$ is the set of jobs and $\mathcal{M}$ is the set of machines. The variable $x_{i,j}$, for each $i\in \mathcal{M}$ and $j\in \mathcal{J}$, indicates whether job $j$ is scheduled on machine $i$. Note that if there is no feasible solution to the relaxation, then the algorithm declares there is no schedule with makespan at most $T$.

%

\begin{figure} 

		\begin{small}
			\arraycolsep=1.0pt\def\arraystretch{1.8}
			$\begin{array}{lllll}
			\text{(\LPgap)} &&&&\\
		\text{min}& \displaystyle\sum_{j\in \mathcal{J}} \sum_{i\in \mathcal{M}} x_{i,j}c_{i,j}&&&\\
			\text{s.t}& \displaystyle \sum_{i\in \mathcal{M}} x_{i,j} = 1 &&\forall j\in \mathcal{J}&(Job)\\
			& \displaystyle\sum\limits_{j\in \mathcal{J}} x_{i,j}p_{i,j} \leq 1 &&\forall i\in \mathcal{M}&(Load) \\
			& x_{i,j}=0 && \forall i\in \mathcal{M}, j\in \mathcal{J}: p_{i,j}>1&\\
			& x_{i,j} \geq 0 && \forall i\in \mathcal{M}, j\in \mathcal{J}&\\
			\end{array}$
		\end{small}

	\caption{The relaxation by Shmoys and Trados \cite{ShmoysT93} to \sGAP.}
\end{figure} \label{RelaxationST}

Given a solution \textbf{x} to \LPgap, the algorithm of \cite{ShmoysT93} constructs a weighted bipartite graph $G=(\mathcal{J},S, E)$, which will be described shortly. Afterwards, the algorithm finds a minimum cost perfect matching to the side $\mathcal{J}$, {\em i.e.}, each vertex in $\mathcal{J}$ is matched to a vertex in $S$. Using this matching the algorithm assigns each job to a machine.
The  bipartite graph $G$ is constructed as follows, where we assume that $\mathcal{J}=\{1,2,\ldots, n\}$ is the set of jobs and $S$ is a collection of ``slots''. Machine $i$ is allocated $k_i \triangleq \lceil \sum_{j=1}^{n} x_{i,j} \rceil$ slots which we denote by $slot(i,1),\ldots slot(i,k_i)$, each having a capacity of $1$. 
For each machine $i$ sort the jobs in a non-increasing order of their processing time $p_{i,j}$, and for each job $j$ in this order add $x_{i,j}$ units of job $j$ to the next non-full slot of machine $i$ (starting from $slot(i,1)$).
If $x_{i,j}$ is larger than the remaining capacity of the slot, which we denote by $r$, add $r$ units of job $j$ to that slot and $x_{i,j}-r$ units of job $j$ to the next slot. 
An edge connecting job $j$ and a slot $(i,\ell)$ is added to $E$ if some of the $x_{i,j}$ units of $j$ were added to the slot $(i,\ell)$, and its cost is set to $c_{i,j}$.
A description of \cite{ShmoysT93} appears in Algorithm \ref{STrounding}.

\begin{algorithm*}[H] \label{STrounding}
	\caption{Shmoys-Tardos $(\textbf{x}, \textbf{p}, \textbf{c})$}
	Construct the bipartite graph $G=(\mathcal{J},S,E)$ as described above. \\
	Find in $G$ a minimum cost perfect matching with respect to $\mathcal{J}$. \\
	For each job $j\in \mathcal{J}$, assign $j$ to machine $i$ if the slot that is matched to $j$ belongs to $i$. 
\end{algorithm*}

We say a slot is \textit{full} if the remaining capacity of that slot is $0$.
Additionally, we say a job $j$ is on {\em top} of a slot if $j$ is the first job to be inserted to that slot. 
It can be proved that the load on machine $i$ in the output of Algorithm \ref{STrounding} is at most $1+p_{i,1}$, where $p_{i,1}$ is the processing time of the job on top of $slot(i,1)$, \ie, the largest processing time of a job that is fractionally scheduled on machine $i$. Since, $p_{i,1}\leq 1$, the makespan of the assignment is at most $2$. Furthermore, it can be shown that the cost of the assignment is at most $OPT_{\text{\LPgap}}$, and thus at most $C(T)$.

We remark that when one is aiming to solve the single criterion version of this problem, \ie, finding an assignment that minimizes the makespan, a binary search could be preformed to find the smallest $T$ such that the linear relaxation is feasible.
In general, any $(\alpha,\beta)$-approximation for the bicriteria problem implies an approximation of $\alpha$ for the single criteria problem.

%% file: Algorithm.tex
We start by describing the general framework in the setting of \sGB.
For simplicity of presentation, given a target makespan $T$, if there exists an edge $e$ such that $p_e>T$ the algorithm immediately declares that there is no orientation with makespan at most $T$.
Otherwise, we scale the processing times by $T$.
Thus, without loss of generality, $T=1$ and $p_e\leq 1$ for every $e\in E$.

Currently, we consider the relaxation of \cite{EbenlendrKS08}, which we denote by \LPthree, with the addition of an objective function that minimizes the orientation cost.\footnote{In Section \ref{Section:Extensions} we also need the constraint that appears in the relaxation of \cite{ShmoysT93} which states that $x_{e,u}=0$ if $p_{e,u}>1$, for every $e\in E$ and $u\in e$.}
This relaxation appears in Figure \ref{RelaxationGB}.

\begin{figure} 
	\begin{small}
		\arraycolsep=1.0pt\def\arraystretch{1.8}
		$\begin{array}{lllll}
		\text{(\LPthree)} &&&&\\
		\text{min}& \displaystyle\sum_{e\in E} \sum_{u\in e} x_{e,u}c_{e,u}&&&\\
		\text{s.t}& \displaystyle \sum_{u\in e} x_{e,u} = 1 &&\forall e\in E&(Edge)\\
		& \displaystyle\sum\limits_{e\in \delta(u)} x_{e,u}p_e \leq 1 &&\forall u\in V&(Load) \\
		& \displaystyle\sum_{e\in \delta(u):\  p_e>0.5} x_{e,u} \leq 1 &&\forall u\in V&(Star) \\
		& x_{e,u} \geq 0 && \forall u\in V, e\in \delta(u)&\\
		\end{array}$
	\end{small}
	\caption{The relaxation by Ebenlendr {\em et al.} \cite{EbenlendrKS08} to \sGB.}
\end{figure} \label{RelaxationGB}

Note that the Star constraint of \LPthree implies that at most a total fraction of $1$ of {\em big} edges, {\em i.e.}, edges whose weight is larger than $\nicefrac[]{1}{2}$, can be oriented toward $u$.
Moreover, we note that later we strengthen this relaxation by adding additional constraints.

Once the processing times are scaled by $T$, the algorithm solves the relaxation \LPthree.
If there is no feasible solution to the relaxation, then the algorithm declares that there is no orientation with makespan at most $T$.
Thus, from this point onward we assume that \LPthree is feasible and focus on the rounding.


Recall that the rounding consists of only two steps, the first local and the second global.
In the first step, some of the edges might be oriented, where an edge $e$ is oriented toward $u$ if $x_{e,u}> f(p_e)$ for a given threshold function $ f:[0,1]\rightarrow [\nicefrac[]{1}{2},1]$.
We employ the framework for threshold functions $f$ which are monotone non-increasing, thus making lighter edges less likely to be oriented compared to heavier edges.
In the second step, the remaining un-oriented edges are oriented using Algorithm \ref{STrounding}.
The framework is described in Algorithm \ref{GBrounding}. It receives as an input: (1) the graph $G=(V,E, \textbf{p},\textbf{c})$; (2) \textbf{x} a solution to the relaxation; (3) a threshold function $f$. 

\begin{algorithm*}[H] \label{GBrounding}
	\caption{Framework$(G=(V,E,\textbf{p},\textbf{c}), \textbf{x}, f)$}
	For each edge $e$ and $u\in e$: if $x_{e,u} > f(p_e)$ then orient $e$ to $u$ and remove $e$ from $E$. \hfill (\LS) \\
	Execute Algorithm \ref{STrounding}. \hfill (\GS)
\end{algorithm*}

Note that the \LS of Algorithm \ref{GBrounding} is well defined, {\em i.e.}, no edge is oriented to both its endpoints. This is due to the Edge constraints and the fact that for each $p\in [0,1]$: $f(p)\geq \nicefrac{1}{2}$. Note that the framework captures Algorithm \ref{STrounding} as a special case since one can choose $f \equiv 1$. We now focus on bounding the makespan and orientation cost produced by the framework, for a general threshold function $f$. This analysis will be useful for the rest of the paper.


\vspace{3pt}
\noindent {\bf Makespan:} We start by presenting a simple but crucial observation. The observation states that if an edge $e=(u,v)$ was \textit{not} oriented at the \LS then $x_{e,u}$ and $x_{e,v}$ cannot vary much. It is important to note that this is the \textit{only} place in our proof we use the fact that $e$ is an edge, {\em i.e.}, the job that corresponds to $e$ can be assigned to only two machines $u$ and $v$ (otherwise our algorithm could have been applied to the more general problem of \sRA).

\begin{observation} \label{GraphObserv}
Let $e=(u,v)\in E$ such that $e$ was not oriented to either $u$ or $v$ in the \LS.
Then $1-f(p_e)\leq x_{e,u}\leq f(p_e)$.
\end{observation}

\begin{proof}
$e$ was not oriented toward $u$ in the \LS, and therefore $x_{e,u}\leq f(p_e)$.
Additionally, the Edge constraint implies that $x_{e,v} = 1-x_{e,u}$, and since $e$ was not oriented toward $v$ in the \LS then $1-x_{e,u}\leq f(p_e)$. This concludes the proof. 
\end{proof}


Now we focus on bounding the makespan. Fix a vertex $u\in V$, and denote the slots that were allocated to $u$ in Algorithm \ref{STrounding} by: $slot(u,1),...,slot(u,k)$ or alternatively by $s_1,...,s_k$. For $i\in \{1, 2, \ldots, k\}$ let $e_i$ be the edge on top of $slot(u,i)$ and denote its processing time by $p_i$. We assume without loss of generality that $p_{k+1}\triangleq 0$ and $x_{e_{k+1},u}=1$ (one can simply add a $0$ weight edge that is fully oriented toward $u$).\footnote{Alternatively, we can also assume $p_{k+2}=0$ and $x_{e_{k+2},u}=1$ as well.} Additionally, denote by $e'_1,\dots , e'_t$ the edges that were oriented to $u$ in the \LS, and denote by $q_1,\dots , q_t$, their processing times respectively.
Lastly, for a slot $s$ and edge $e$ we denote by $y_{e,s}$ the fraction that $e$ is assigned to $s$.

We now introduce a new observation that lower bounds the fractional load in the first slot of $u$, {\em i.e.}, $\sum _{e\in slot(u,1)}y_{e,s_1}p_e$.
This observation will be useful in bounding the load on $u$.

\begin{observation} \label{FirstSlotObserv}
	The fractional load in the first slot of $u$ is at least:
	\[
	\sum_{e\in slot(u,1)}y_{e,s_1}p_e\geq (1-f(p_1))p_1 + f(p_1)p_2.
	\]
\end{observation}
\begin{proof}
	From Observation \ref{GraphObserv} we know that $x_{e_1}\geq 1-f(p_1)$. Moreover, since $e_1$ is the first edge to be inserted to the first slot, then it is contained fully in $slot(u,1)$. Therefore, $y_{e_1,s_1}=x_{e_1,u}\geq 1-f(p_1)$. Recall that $p_2\leq p_1$. Since $slot(u,1)$ is full and its capacity equals $1$, we can conclude: $ \sum_{e\in slot(u,1)}y_{e,s_1}p_e\geq f(p_1)p_1 + (1-f(p_1))p_2$.
\end{proof}
Now we introduce a lemma that is inspired by \cite{ShmoysT93} and upper bounds the load on $u$.
\begin{lemma} \label{lemmaST}
	Let $e'_1,\dots, e'_t$ be the edges that were oriented to $u$ in the \LS, and let $q_1,\dots, q_t$ be their processing times respectively. Then, 
	\[
	\sum_{i=1}^{t}q_i+\sum_{i=1}^{k}p_i\leq 1+\sum_{i=1}^{t}(1-f(q_i))q_i+f(p_1)p_1+(1-f(p_1))p_2.
	\]
\end{lemma}

\begin{proof}
	First, recall that for every $1\leq s \leq k-1$ $slot(u,s)$ has a capacity exactly $1$. Moreover, the slots are filled with edges in decreasing order of processing time. Therefore, we can deduce that for each $1\leq i\leq k-1$:
	\begin{align*}
	\sum_{e\in slot(u,i)} y_{e,s_i}p_e &\geq
	\sum_{e\in slot(u,i)} y_{e,s_i}p_{i+1}=p_{i+1}\sum_{e\in slot(u,i)} y_{e,s_i}=p_{i+1}.
	\end{align*}
Since at most one edge from each slot can be selected in the \GS, the load on $u$ from edges that are oriented to $u$ in the \GS is at most $\sum_{i=1}^{k}p_i$. From the above inequality, along with Observation \ref{FirstSlotObserv}, we can conclude that:
	\begin{align*}
	\sum_{i=1}^{t}q_i+\sum_{i=1}^{k}p_i &= \sum_{i=1}^{t}q_i+p_1+p_2+\sum_{i=3}^{k}p_i
	\leq \sum_{i=1}^{t}q_i+p_1+p_2+\sum_{i=2}^{k-1}\sum_{e\in slot(u,i)} y_{e,s_i}p_e \\
	&\leq \sum_{i=1}^{t}q_i+p_1+p_2+\sum_{i=1}^{k}\sum_{e\in slot(u,i)} y_{e,s_i}p_e - \sum_{e\in slot(u,1)} y_{e,s_1}p_e\\
	&\leq \sum_{i=1}^{t}q_i+p_1+p_2+\sum_{e\in \delta(u)} x_{e,u}p_e - \left( (1-f(p_1))p_1+f(p_1)p_2 \right)\\
	&\leq \sum_{i=1}^{t}q_i+f(p_1)p_1+(1-f(p_1))p_2+\left( 1-\sum_{i=1}^{t}f(q_i)q_i \right)\\
	&=1+\sum_{i=1}^{t}(1-f(q_i))q_i+f(p_1)p_1+(1-f(p_1))p_2.
	\end{align*}
The last inequality follows from the Load constraint on $u$, and the fact that the edges $e'_1,\dots, e'_t$ were removed from $E$ at the end of the \LS.
\end{proof}

Lastly, we observe that all of the {{\em big} edges, {\em i.e.}, edges whose weight is larger than $\nicefrac[]{1}{2}$, that are not oriented toward $u$ in the \LS are assigned to the first slot. 
This is summarized in the following observation.
\begin{observation} \label{SecondSlotObserv}
	Let $e$ be an edge in $slot(u,i)$ such that $i>1$. Then, $p_e\leq \nicefrac{1}{2}$.
\end{observation}
\begin{proof}
	Assume for the sake of contradiction that $p_e>\nicefrac{1}{2}$. 
Since the slots are filled in a non-increasing weight order, all edges in slots $1, 2, \dots, i-1$ are filled with fractions of edges whose processing time is greater than $\nicefrac{1}{2}$. Therefore, $\sum_{e\in \delta(u): \  p_e>\nicefrac{1}{2}}x_{e,u}>1$, which contradicts the Star constraint on $u$.
\end{proof}

\vspace{3pt}
\noindent {\bf Orientation Cost:} The following lemma upper bounds the orientation cost of the orientation produced by Algorithm \ref{GBrounding}.

\begin{lemma} \label{CostApprox}
	Given $f:[0,1]\rightarrow [\nicefrac{1}{2}, 1]$, let $c \triangleq (inf\{f(p) | p\in [0,1]\})^{-1}$. Then Algorithm \ref{GBrounding} with $f$ outputs an orientation with a cost of at most $c\cdot C(T)$.
\end{lemma}
\begin{proof}
	Let $S$ be the set of edges that were oriented to $u$ in the \LS of Algorithm \ref{GBrounding}. Let $\eta$ be the orientation the algorithm outputs, {\em i.e.}, $\eta(e)$ equals the vertex that $e$ is oriented to. Since Algorithm \ref{STrounding} does not lose in the assignment cost:
	\[
	\sum_{e\in E\setminus S} c_{e,\eta(e)} \leq \sum_{e=(u,v)\in E\setminus S} (c_{e,u} x_{e,u} + c_{e,v} x_{e,v}).
	\]
	
	For each $e\in S$: $x_{e,\eta(e)}>f(p_e)\geq \frac{1}{c}$. So we deduce: $c\cdot x_{e,\eta(e)}\geq 1$. Therefore, $c_{e,\eta(e)} \leq c\cdot x_{e,\eta(e)}c_{e,\eta(e)}$.
	Then we can conclude:
	
	\[
	\sum_{e\in S} c_{e,\eta(e)} \leq c\sum_{e=(u,v)\in S} (c_{e,u} x_{e,u} + c_{e,v} x_{e,v}).
	\]
	Since $1\leq c \leq 2$, we can deduce:
	\[
	\sum_{e\in E} c_{e,\eta(e)} \leq c\sum_{e=(u,v)\in E} (c_{e,u} x_{e,u} + c_{e,v} x_{e,v})\leq c\cdot OPT_{LP},
	\]
	where $OPT_{LP}$ is the optimal value of feasible solution to the relaxation. This implies that the orientation outputted by Algorithm \ref{GBrounding} with a rounding function $f$, has a cost of at most ${c\cdot C(T)}$.
\end{proof}

%% file: Analysis.tex

Let us now focus on applying the framework, with an appropriate threshold function $f$, to \sGB.
First, we present a theorem that achieves part of the tradeoff claimed in Theorem \ref{Theorem:firstUpperBound}, and only in the next subsection we show how to extend this tradeoff to fully achieve Theorem \ref{Theorem:firstUpperBound}.

\begin{theorem} \label{GBBasicTheorem}
There exists a threshold function $f$ such that Algorithm \ref{GBrounding} finds an orientation that is a $(1.75+\gamma, \nicefrac{1}{(2\gamma+0.5)})$-approximation, for every $\nicefrac{1}{12}\leq \gamma \leq \nicefrac{1}{4}$.
\end{theorem}
The function $f_\alpha$ we use in the proof of Theorem \ref{GBBasicTheorem} is the following:
\begin{align} \label{ThresholdFunc}
f_\alpha(p_e) =
\begin{cases}
1 & \text{if } p_e\leq \nicefrac{1}{2} \\
\alpha & \text{if } p_e> \nicefrac{1}{2} \\
\end{cases}
\end{align}
where $\nicefrac{2}{3}\leq \alpha \leq 1$. The following lemma upper bounds the makespan of Algorithm \ref{GBrounding} with the above $f_\alpha$.
\begin{lemma} \label{MakespanLemma}
	The makespan of the orientation produced by Algorithm \ref{GBrounding} with $f_\alpha$ is at most:
	$1.5+0.5\alpha$, where $\nicefrac{2}{3}\leq \alpha\leq 1$.
\end{lemma}

\begin{proof}
Consider the number of edges that were oriented toward $u$ in the \LS. First, we note that from the Star constraint on $u$, at most one edge can be oriented toward $u$ in the \LS.
If this is not the case then let $e'_1$ and $e'_2$ be edges oriented to $u$ in the \LS. Then, $p_{e_1},p_{e_2}>\nicefrac{1}{2}$. However, $x_{e_1,u}+x_{e_2,u}>\alpha+\alpha\geq\nicefrac{2}{3}+\nicefrac{2}{3}>1$, which contradicts the Star constraint on $u$.
Hence, there are only two cases to consider.	
	
	
\noindent	\textbf{Case 1:} Assume no edge is oriented toward $u$ in the \LS. Therefore, using Lemma \ref{lemmaST} and Observation \ref{SecondSlotObserv} the load on $u$ is at most:
	\begin{align*}
	\sum_{i=1}^{k}p_i
	& \leq 1+f_\alpha(p_1)p_1+(1-f_\alpha(p_1))p_2
	\leq 1.5+f_\alpha(p_1)(p_1-0.5)\\
	&\leq 1.5+\alpha\cdot (1-0.5)
	= 1.5 + 0.5\alpha,
	\end{align*}
where the last inequality follows from the fact that the expression: $f_{\alpha}(p_1)(p_1-0.5)$ is maximized when $p_1=1$ (and thus $f_\alpha(p_1)=\alpha$).

\noindent	\textbf{Case 2:} Assume there is exactly one edge that was oriented toward $u$ in the \LS. Recall we denote this edge as $e'_1$ and its processing time by $q_1$. 
Since $q_1> \nicefrac{1}{2}$ and $x_{e'_1,u}>\alpha$, then it must be the case that $p_1\leq \nicefrac{1}{2}$ (otherwise Observation \ref{GraphObserv} implies that $x_{e'_1,u}+x_{e_1,u}>\alpha+1-\alpha=1$, which contradicts the Star constraint for $u$).	
	Therefore, from Lemma \ref{lemmaST} the load on $u$ in the output of Algorithm \ref{GBrounding} is at most:
	\begin{align*}
	q_1+\sum_{i=1}^{k}p_i
	&\leq 1+(1-f_\alpha(q_1))q_1+f_\alpha(p_1)p_1+(1-f_\alpha(p_1))p_2
	\leq 1+(1-\alpha)q_1+p_1\\
	&\leq 1+(1-\alpha)+0.5
	=2.5-\alpha
	\leq 1.5+0.5\alpha.
	\end{align*}
	The second inequality follows from the fact that $p_2\leq p_1$ and $f_{\alpha}(q_1)=\alpha$ (since $q_1>\nicefrac[]{1}{2}$), whereas the third inequality from the fact that $p_1\leq \nicefrac[]{1}{2}$. In addition, the last inequality follows from the fact that $\nicefrac{2}{3}\leq \alpha\leq 1$.
\end{proof}
Now, we are ready to conclude the proof of Theorem \ref{GBBasicTheorem}:
\begin{proof}[Proof of Theorem \ref{GBBasicTheorem}]
Applying Lemma \ref{MakespanLemma}, Lemma \ref{CostApprox} and choosing $\gamma = 0.5\alpha-0.25$ finishes the proof.
\end{proof}

\vspace{3pt}
\noindent {\bf Tightness of Analysis:} We now show that the analysis of Algorithm \ref{GBrounding} with a threshold function $f_\alpha$ is tight. Formally, we prove the following lemma.

\begin{lemma} \label{lemmaTightness}
For every $\nicefrac[]{1}{2}\leq \alpha<1$ there exists an instance such that the output of Algorithm \ref{GBrounding} with $f_\alpha$ has makespan at least $max\left\{1.5+0.5\alpha, 2.5-\alpha \right\}$ and orientation cost at least $\nicefrac{1}{\alpha}\cdot OPT_{LP}$.
\end{lemma}

\begin{proof}
	We introduce two instances, both of them feasible to \LPthree. The output of Algorithm \ref{GBrounding} on the first instance has makespan of $1.5+0.5\alpha$ and orientation cost of $\nicefrac{1}{\alpha}\cdot OPT_{LP}$.
	The output of Algorithm \ref{GBrounding} on the second instance has the same orientation cost, however the makespan is $2.5-\alpha$.
	This proves the desired result.
	
	Let $\epsilon>0$. The first instance is shown in Figure \ref{TightnessFigureBoth} on the left and consists of five vertices. The processing time of an edge is written above it. Furthermore, there is a load value, {\em i.e.}, a self loop, on each vertex with a specified weight.
	These load values are denoted by $q$, {\em e.g.}, $q_u$ denotes the load value of vertex $u$.
	We note that in order to ensure the Star constraints of $LP$ are feasible, we split each self loop into two self loops (each with half the original load value).
	In addition to the information in Figure \ref{TightnessFigureBoth}, there are orientation costs to each edge: $c_{(u,v_1),v_1}=c_{(u,v_2),v_2}=\epsilon$ and $c_{(u'v'),v'}=1$. The rest of the orientation costs are $0$.
	
	\begin{center}	
		\setlength{\abovecaptionskip}{2pt plus 2pt minus 2pt}
		\captionsetup[table]{font=small,skip=0pt}     
		\begin{figure}
			\begin{tabular}{@{}c@{}}
				\begin{tikzpicture}
				[scale=.5,auto=left]
				\node[circle,fill=blue!20] (u) at (4,7)  {$u$};
				\node[circle,fill=blue!20] (v1) at (8,9)  {$v_1$};
				\node[circle,fill=blue!20] (v2) at (8,5)  {$v_2$};
				\node[circle,fill=blue!20] (u') at (12,9)  {$u'$};
				\node[circle,fill=blue!20] (v') at (12,5) {$v'$};
				
				\node[circle] (qu) at (3,6) {$q_u=0.5\alpha-0.5\epsilon$};
				\node[circle] (qv1) at (7,10) {$q_{v_1}=1-\alpha$};
				\node[circle] (qv2) at (6,4) {$q_{v_2}=0.5+0.5\alpha+0.5\epsilon$};
				\node[circle] (qu') at (12,10) {$q_{u'}=1-\alpha-\epsilon$};
				\node[circle] (qv') at (12,4) {$q_{v'}=\alpha+\epsilon$};
				
				\draw[draw=black, line width=0.5] (u) edge node{1} (v1);
				\draw[draw=black, line width=0.5] (u) edge node{$\nicefrac{1}{2}$} (v2);
				\draw[draw=black, line width=0.5] (u') edge node{1} (v');
				
				\end{tikzpicture}
				\\[\abovecaptionskip]
				\small {(a) Instance with gap of} {$1.5+0.5\alpha$} 
				
			\end{tabular}
			\begin{tabular}{@{}c@{}}
				\begin{tikzpicture} 
				[scale=.6,auto=left]
				\node[circle,fill=blue!20] (u) at (4,7)  {$u$};
				\node[circle,fill=blue!20] (v1) at (8,9)  {$v_1$};
				\node[circle,fill=blue!20] (v2) at (8,5)  {$v_2$};
				
				\node[circle] (qu) at (3,6) {$q_u=1-\alpha-\epsilon$};
				\node[circle] (qv1) at (8,10) {$q_{v_1}=\alpha+\nicefrac{\epsilon}{2}$};
				\node[circle] (qv2) at (8,4) {$q_{v_2}=0.5+0.5\epsilon$};
				
				\draw[draw=black, line width=0.5] (u) edge node{1} (v1);
				\draw[draw=black, line width=0.5] (u) edge node{$\nicefrac{1}{2}$} (v2);	
				\end{tikzpicture}
				
				\\[\abovecaptionskip]
				\small {(b) Instance with gap of } {$2.5-\alpha$} 
				
			\end{tabular}
			\caption{Gap instances for Lemma \ref{lemmaTightness}.}
			\label{TightnessFigureBoth}
		\end{figure}
		
	\end{center}
	
	The only feasible solution to \LPthree is $x_{(u,v_1),u}=1-\alpha$, $x_{(u,v_2),u}=\alpha+\epsilon$ and $x_{(u',v'),u'}=\alpha+\epsilon$. The reason that this is the only feasible solution is that the fractional load on each vertex is exactly $1$ and there are no cycles in the graph. Moreover, the fractional cost of the orientation is:
	\begin{align*}
	OPT_{LP}&=\alpha+\epsilon+\epsilon(\alpha+1-\alpha-\epsilon)\\
	&\leq \alpha+2\epsilon .
	\end{align*}
	
	Algorithm \ref{GBrounding} with threshold function $f_\alpha$ will orient $(u',v')$ toward $u'$ in the \LS. Moreover, since Algorithm \ref{STrounding} finds minimum cost matching the edges $(u,v_1)$ and $(u,v_2)$ are oriented toward $u$ in the \GS. This is due to the fact that the cost of the fractional orientation is $\epsilon(\alpha+1-\alpha-\epsilon)=\epsilon(1-\epsilon)$, so orienting an edge toward $v_1$ or $v_2$ will incur a cost of at least $\epsilon$. In conclusion, the output of Algorithm \ref{GBrounding} with $f_\alpha$ as a threshold function, has makespan of $1.5+0.5\alpha-0.5\epsilon$. Additionally, the orientation cost is $1=\nicefrac{1}{(\alpha+\epsilon)}OPT_{LP}$.

	The second instance is shown in Figure \ref{TightnessFigureBoth} on the right and consists of three vertices. Similar to the previous instance, the processing time of an edge is written above it, and the load value of a vertex is written next to it. In addition, the orientation costs are $c_{(u,v_1),v_1}=1$ and $c_{(u,v_2),v_2}=\epsilon$ (all other orientation costs are $0$).

	The only feasible solution to \LPthree is $x_{(u,v),u}=\alpha+0.5\epsilon$ and $x_{(u,v_2),u}=\epsilon$. Note that in this case the load on each vertex is exactly $1$. Since there are no cycles in this graph this is the only feasible solution to $LP$. The fractional orientation cost is:
	\[
	OPT_{LP}=\alpha+0.5\epsilon+\epsilon(0.5+0.5\epsilon).
	\]
	Algorithm \ref{GBrounding} with threshold function $f_\alpha$, will orient the edge $(u,v_1)$ toward $u$ in the \LS. In addition, since Algorithm \ref{STrounding} finds a minimum cost matching, $(u,v_2)$ is oriented toward $u$ in the \GS. To conclude, the makespan is $2.5\alpha-\epsilon$ and the orientation cost is $1=\nicefrac{1}{(\alpha+\epsilon(1+0.5\epsilon))}OPT_{LP}$.
	This concludes the proof.
\end{proof}

Lemma \ref{lemmaTightness} shows the analysis of Algorithm \ref{GBrounding} with a threshold function $f_\alpha$ is tight. Consequently, in order to extend the bicriteria tradeoff of Theorem \ref{GBBasicTheorem}, and obtain Theorem \ref{Theorem:firstUpperBound}, we require a different threshold function and a stronger relaxation.

%% file: ImprovedGB.tex
It is important to note that Lemma \ref{lemmaTightness} implies that using Algorithm \ref{GBrounding} with \LPthree and the threshold function $f_\alpha$ cannot achieve an approximation better than $\nicefrac[]{11}{6}$ with respect to the makespan.
To this end we strengthen \LPthree using the following constraint (which we denote by Set constraints):

\begin{align*}
\displaystyle\sum_{e\in S} x_{e,u} \leq |S|-1 \quad\quad \forall u\in V, \forall S\subseteq \delta(u):S\notin \mathcal{F}(u)\text{ and } |S|\leq k \quad (Set)
\end{align*}
We call the new relaxation \LPstong.\footnote{Similarly to \LPthree, for some of the extensions of \sGB we add that $x_{e,u}=0$ if $p_{e,u}>1$ (for every $e\in E$ and $u\in e$).} 
Intuitively, the Set constraints enforce that given an infeasible set of edges $S$ touching $u$ not all edges of $S$ can be oriented toward $u$.
In fact, for our specific choice of a threshold function $f$ we use $k=3$. Thus, no separation oracle is required when solving the relaxation. The exact result is formulated in the following theorem:


\begin{theorem} \label{BetterGBTheorem}
There exists a rounding function $f$ such that Algorithm \ref{GBrounding} finds an orientation that is a $(1.75+\gamma, \nicefrac{1}{(2\gamma+0.5)})$-approximation, for every $\nicefrac{1}{12}-\nicefrac{\epsilon}{2}\leq \gamma \leq \nicefrac{1}{12}$, where $\epsilon=\nicefrac{\sqrt{33}}{2}-\nicefrac{17}{6}$.
\end{theorem}

Note that this theorem extends the tradeoff achieved in Theorem \ref{GBBasicTheorem}, and together both theorems achieve the tradeoff of Theorem \ref{Theorem:firstUpperBound}. The threshold function $f_\epsilon$ we use in the proof of Theorem \ref{BetterGBTheorem} is defined as follows:

\begin{align} \label{BetterThreshFunc}
f_\epsilon(p_e) =
\begin{cases}
\nicefrac{2}{3}-\epsilon & \text{if } p_e>\nicefrac{1}{2} \\
\nicefrac{2}{3}+\nicefrac{\epsilon}{2} & \text{if } \nicefrac{1}{3} < p_e\leq \nicefrac{1}{2} \\
1 &					\text{if } p_e \leq \nicefrac{1}{3}\\
\end{cases}
\end{align}
where $0\leq \epsilon\leq \nicefrac{\sqrt{33}}{2}-\nicefrac{17}{6}$. 
The following lemma upper bounds the makespan.
%
%

\begin{lemma} \label{BetterMakespanLemma}
	The output of Algorithm \ref{GBrounding} with the threshold function $f_\epsilon$, has a makespan of at most $\nicefrac{11}{6}-\nicefrac{\epsilon}{2}$.
\end{lemma}

Before we proceed to the proof of the previous lemma, we use a simple but crucial observation that follows directly from the Set constraints of our strengthened relaxation $LP_k$.

\begin{observation} \label{SetObserv}
	Let $u\in V$ and $S\subseteq \delta(u)$. If $\sum_{e\in S}x_{e,u}>|S|-1$, then $\sum_{e\in S}p_e\leq 1$.
\end{observation}

\begin{proof}
	Follows immediately from the Set constraint of $u$ and $S$.
\end{proof}
Now we give the proof of Lemma \ref{BetterMakespanLemma}:
\begin{proof}[Proof of Lemma \ref{BetterMakespanLemma}]
	Fix a vertex $u\in V$. 
	We start with the following simple observation:
	\begin{observation} \label{NumOriented}
		The number of edges oriented to $u$ in the \LS of Algorithm \ref{GBrounding} is at most $2$.
	\end{observation}
	\begin{proof}
		Assume for the sake of contradiction that $e'_1,e'_2,e'_3$ were oriented to $u$ in the \LS. First, denote the weight of $e'_i$ by $q_i$. From the Star constraint on $u$, there is at most one edge $e'_i$ (among the three edges) such that $q_i>\nicefrac{1}{2}$. Therefore we can deduce:
		\begin{align*}
		x_{e'_1,u}+x_{e'_2,u}+x_{e'_3,u}
		&>f(p_{e'_1})+f(p_{e'_2})+f(p_{e'_3})\\
		&\geq \frac{2}{3}-\epsilon+\frac{2}{3}+\epsilon+\frac{2}{3}+\epsilon\\
		&>2.
		\end{align*}
		From Observation \ref{SetObserv}, we get that: $q_1+q_2+q_3\leq 1$, but on the other hand from the definition of $f_\epsilon$ (see \ref{BetterThreshFunc}), $q_i>\nicefrac{1}{3}$, for each $i$, which is a contradiction.
	\end{proof}
	
	Therefore we denote by $e'_1,e'_2$ the edges oriented to $u$ in the \LS, and their weights by $q_1,q_2$. If there is one (or no) edges oriented to $u$ in the \LS, then both $x_{e'_2,u}$ and $q_2$ equal $0$ (or all of $x_{e'_1,u},q_1,x_{e'_2,u},q_2$ equal $0$) respectively.
	
	The second observation states the sum $q_1+q_2$ cannot be too large. Formally:
	\begin{observation} \label{SumRounded}
		It must be the case that $q_1+q_2\leq 1$.
	\end{observation}
	\begin{proof}
		If $q_2=0$ or $q_1=q_2=0$ then trivially, $q_1+q_2\leq 1$. Otherwise, $q_1,q_2>0$, {\em i.e.}, both $e'_1,e'_2$ are edges that were oriented to $u$ in the \LS. Therefore,
		\[
		x_{e'_1,u}+x_{e'_2,u}>f_\epsilon(q_1)+f_\epsilon(q_2)>\nicefrac{1}{2}+\nicefrac{1}{2}=1.
		\]
		From Observation \ref{SetObserv} (choosing $S=\{e'_1,e'_2\}$), we conclude: $q_1+q_2\leq 1$.
	\end{proof}
	
	We now consider three cases regarding the value of $p_1$ and bound the load on $u$ in each case concisely (the bound we achieve is $\nicefrac{11}{6}-\nicefrac{\epsilon}{2}$).
	
	\noindent \textbf{Case 1:} In this case $p_1\leq \nicefrac{1}{3}$. From Lemma \ref{lemmaST} we get that the load on $u$ is at most:
	\begin{align*}
	1+(1-f_\epsilon(q_1))q_1+(1-f_\epsilon(q_2))q_2+f_\epsilon(p_1)p_1+&(1-f_\epsilon(p_1))p_2\leq \\
	&\leq 1+\left(\frac{1}{3}+\epsilon\right)(q_1+q_2)+p_1\\
	&\leq 1+\left(\frac{1}{3}+\epsilon\right)+\frac{1}{3}\\
	&=\frac{5}{3}+\epsilon\\
	&<\frac{11}{6}-\frac{\epsilon}{2}.
	\end{align*}
	The first inequality follows from the definition of $f_\epsilon$ (and the fact that $p_2\leq p_1$). Additionally, the second inequality follows from Observation \ref{SumRounded}.
	
	
	\noindent \textbf{Case 2:} In this case $p_1>\nicefrac{1}{2}$. Note that it must hold that $q_1,q_2\leq \nicefrac{1}{2}$ (otherwise Observation \ref{GraphObserv} implies that $x_{e_1,u}+x_{e'_i,u}>1-f_\epsilon(p_1)+f_\epsilon(p_{e'_i})=1-f_\epsilon(p_1)+f_\epsilon(p_1)=1$, which contradicts the Star constraint on $u$). Hence, $\nicefrac{1}{3}<q_1,q_2\leq \nicefrac{1}{2}$.
	
	Moreover it holds that $p_1+q_i\leq 1$ (Observation \ref{SetObserv} applied to the set $S=\{e'_i,e_1\}$ along with Observation $\ref{GraphObserv}$). From Lemma \ref{lemmaST} we derive that the load on $u$ is at most:
	
	\begin{align*}
	&1+(1-f_\epsilon(q_1))q_1+(1-f_\epsilon(q_2))q_2+f_\epsilon(p_1)p_1+(1-f_\epsilon(p_1))p_2\\
	&= 1+\left(\frac{1}{3}-\frac{\epsilon}{2}\right)(q_1+q_2)+\left(\frac{2}{3}-\epsilon\right)p_1+\left(\frac{1}{3}+\epsilon\right)p_2\\
	&= 1+ \left(\frac{1}{3}-\frac{\epsilon}{2}\right)(q_1+p_1)+\left(\frac{1}{3}-\frac{\epsilon}{2}\right)(q_2+p_1)+\left(\frac{1}{3}+\epsilon\right)p_2\\
	&\leq 1+ \left(\frac{1}{3}-\frac{\epsilon}{2}\right)+\left(\frac{1}{3}-\frac{\epsilon}{2}\right)+\left(\frac{1}{3}+\epsilon\right)\cdot \frac{1}{2}\\
	&= 1+ \frac{2}{3}-\epsilon+\frac{1}{6}+\frac{\epsilon}{2}\\
	&=\frac{11}{6}-\frac{\epsilon}{2}.
	\end{align*}
	Note that the second equality follows since $\epsilon\geq 0$.
	
	\noindent \textbf{Case 3:} In this case, $\nicefrac{1}{3}<p_1\leq \nicefrac{1}{2}$. From Lemma \ref{lemmaST} we can bound the load on $u$ by:
	\begin{align*}
	1+(1-f_\epsilon(q_1))q_1+(1-f_\epsilon(q_2))q_2+f_\epsilon(p_1)p_1+(1-f_\epsilon(p_1))p_2\\
	\leq 1 + \left( \frac{1}{3}+\epsilon \right)(q_1+q_2)+\left( \frac{2}{3}+\frac{\epsilon}{2} \right)p_1+ \left( \frac{1}{3}-\frac{\epsilon}{2} \right)p_2.
	\end{align*}
	Instead of using Observation \ref{SecondSlotObserv} we use a different bound on $p_2$ that comes from the Load constraint on $u$ and from Observation \ref{FirstSlotObserv}:
\begin{align*}
\left( \frac{2}{3}-\epsilon \right)(q_1+q_2)+\left( \frac{1}{3}-\frac{\epsilon}{2} \right)p_1+ \left( \frac{2}{3}+\frac{\epsilon}{2} \right)p_2 &\leq x_{e'_1,u}q_1+x_{e'_2}q_2+\sum_{e\in slot(u,1)}y_{e,s_1}p_e \\
& \leq \sum_{e\in \delta(u)}x_{e,u}p_e\leq 1.
\end{align*}
	Therefore, we can deduce:
	\[
	p_2\leq \left( \frac{2}{3}+\frac{\epsilon}{2} \right)^{-1}\left( 1- \left( \frac{2}{3}-\epsilon \right)(q_1+q_2)-\left( \frac{1}{3}-\frac{\epsilon}{2} \right)p_1\right).
	\]
	This bound on $p_2$ and Lemma \ref{lemmaST} give us the following bound on the load on $u$:
	\begin{align*}
	&1 + \left( \frac{1}{3}+\epsilon \right)(q_1+q_2)+\left( \frac{2}{3}+\frac{\epsilon}{2} \right)p_1+\\ &\left( \frac{1}{3}-\frac{\epsilon}{2} \right)\cdot \left( \frac{2}{3}+\frac{\epsilon}{2} \right)^{-1}\left( 1- \left( \frac{2}{3}-\epsilon \right)(q_1+q_2)-\left( \frac{1}{3}-\frac{\epsilon}{2} \right)p_1\right)\\
	&=1+\frac{\frac{1}{3}-\frac{\epsilon}{2}}{\frac{2}{3}+\frac{\epsilon}{2}} + \left( \frac{1}{3}+\epsilon-\frac{(\frac{1}{3}-\frac{\epsilon}{2})(\frac{2}{3}-\epsilon)}{\frac{2}{3}+\frac{\epsilon}{2}} \right)(q_1+q_2)+ \left( \frac{2}{3}+\frac{\epsilon}{2}-\frac{(\frac{1}{3}-\frac{\epsilon}{2})^2}{\frac{2}{3}+\frac{\epsilon}{2}} \right)p_1.
	\end{align*}
	It can be easily be verified that the coefficients of $(q_1+q_2)$ and $p_1$ are non-negative for the range of $\epsilon$.
	Thus, the above expression is maximized when $p_1=\nicefrac{1}{2}$ and $q_1+q_2=1$. Moreover, it can be easily verified that in that case the above expression (when setting $q_1+q_2=1$ and $p_1=\nicefrac{1}{2}$) equals:
	
	\begin{align*}
	1\frac{2}{3}+\frac{5}{4}\epsilon-\frac{\frac{5}{8}\epsilon^2 - \frac{1}{3}\epsilon-\frac{1}{18}}{\frac{2}{3}+\frac{\epsilon}{2}}.
	\end{align*}
	Recalling that $\epsilon\leq \nicefrac{\sqrt{33}}{2}-\nicefrac{17}{6}$ implies that the above is at most: $\nicefrac{11}{6}-\nicefrac{\epsilon}{2}$, concluding the proof.
\end{proof}

Now we conclude with the proofs of Theorems \ref{BetterGBTheorem} and \ref{Theorem:firstUpperBound}.
\begin{proof}[Proof of Theorem \ref{BetterGBTheorem}]
Follows immediately from Lemmas \ref{BetterMakespanLemma} and \ref{CostApprox}, and choosing $\gamma=\nicefrac{1}{12}-\nicefrac{\epsilon}{2}$.
\end{proof}

\begin{proof}[Proof of Theorem \ref{Theorem:firstUpperBound}]
Follows immediately from Theorems \ref{GBBasicTheorem} and \ref{BetterGBTheorem}.
\end{proof}

%% file: NaiveGB.tex
In this section we present a simple rounding algorithm to \LPthree. This algorithm achieves an approximation of $(1.75+\gamma, 1+\nicefrac{1}{\gamma})$ to \sGB, where $0< \gamma\leq 1$ (The exact result is formulated in Theorem \ref{Theorem:secondUpperBound}).

Our approach is similar to the one of Algorithm \ref{GBrounding}, which consists of two steps.
The first local and second global.
Similarly to  Algorithm \ref{GBrounding}, in the local step if an edge is significantly fractionally oriented to a vertex we orient that edge to that vertex. However, there are two differences from Algorithm \ref{GBrounding}. First, the threshold function is constant, \ie $f(p)\equiv \alpha$. Second, after the first \LS, instead of Algorithm \ref{STrounding}, we use the rounding algorithm by Ebenlendr {\em et al.} \cite{EbenlendrKS08}. This is summarized in the following algorithm.

\begin{algorithm*}[H] \label{NGBrounding}
	\caption{$(G=(V,E,\textbf{p},\textbf{c}), \textbf{x},\alpha)$}
	For each edge $e$ and $u\in e$ if $x_{e,u} > \alpha$, then orient $e$ to $u$ and remove $e$ from $E$. \hfill (\LS)\\
	Execute the algorithm of  Ebenlendr {\em et al.} \cite{EbenlendrKS08}. \hfill (\GS)\\
\end{algorithm*}

We now elaborate on the reasons for choosing a constant threshold function to Algorithm \ref{NGBrounding}.
Note that in order to have a bounded loss in the orientation cost, for every $p$, it must hold that $f(p)<1$. Otherwise, if there exists a value $p$ such that $f(p)=1$, then an edge $e$ of weight $p$ is never rounded in the \LS of Algorithm \ref{NGBrounding}.
Therefore, in the \GS it might be oriented to a vertex $u$ in a Leaf Assignment rule (as defined in \cite{EbenlendrKS08}). It is also possible that $x_{e,u}=\epsilon$ for $\epsilon>0$ arbitrarily small. Therefore, the Leaf Assignment step might incur an unbounded loss in the orientation cost.

Thus, we limit the threshold function to a constant function in order to make the analysis simpler. 
We present two lemmas that bound the makespan and orientation cost of Algorithm \ref{NGBrounding}. The first lemma bounds the makespan of the output of Algorithm \ref{NGBrounding}.

\begin{lemma} \label{NGBMakespan}
	The makespan of the orientation Algorithm \ref{NGBrounding}  outputs is at most $\nicefrac{1}{\alpha}+0.75$.
\end{lemma}

The second lemma bound the orientation cost of the output of Algorithm \ref{NGBrounding}:
\begin{lemma} \label{NGBcosts}
The orientation cost of the orientation Algorithm \ref{NGBrounding} outputs is at most $\nicefrac{1}{(1-\alpha)}\cdot C(T)$.
\end{lemma}

From the two lemmas above, the proof of Theorem \ref{Theorem:secondUpperBound} follows immediately:

\begin{proof}[Proof of Theorem \ref{Theorem:secondUpperBound}]
Follows immediately from Lemmas \ref{NGBMakespan}, \ref{NGBcosts} and choosing $\alpha=\nicefrac{1}{(1+\gamma)}$.
\end{proof}

We now proceed to the proofs of Lemmas \ref{NGBMakespan}, \ref{NGBcosts}.

\begin{proof}[Proof of Lemma \ref{NGBMakespan}]
Fix a vertex $u$. Denote by $X$ the fractional load of edges that were oriented to $u$ in the \LS:
\[
X\triangleq \sum_{e\in E_{orient}}x_{e,u}p_e,
\]
where $E_{orient}$ is the set of edges that were oriented to $u$ in the \LS. The load on $u$ from edges oriented in the \LS of Algorithm \ref{NGBrounding} is:
\begin{align*}
\sum_{e\in E_{orient}}p_e &\leq \sum_{e\in E_{orient}}\frac{1}{\alpha}x_{e,u}p_e\\
& = \frac{1}{\alpha}\sum_{e\in E_{orient}}x_{e,u}p_e\\
& =\frac{1}{\alpha}X.
\end{align*}
The first inequality is derived from the definition of the \LS. In addition, the load on $u$ from edges oriented in the \GS is at most: $1-X+0.75$. Note that the algorithm by Ebenlendr {\em et al.} \cite{EbenlendrKS08} adds at most $0.75$ to the initial fractional load of each vertex (see the proof in \cite{EbenlendrKS08}). In this case the initial load is $1-X$, since the initial fractional load equals $1$ and a total of $X$ was removed from $E$ in the \LS of Algorithm \ref{NGBrounding}.

Adding the load on $u$ from edges oriented in the \LS and the \GS of Algorithm \ref{NGBrounding} yields:
\begin{align*}
\frac{1}{\alpha}X+(1-X+0.75)&=1.75+\left(\frac{1}{\alpha}-1\right)X\\
&\leq 1.75+\left(\frac{1}{\alpha}-1\right)\\
&=\frac{1}{\alpha}+0.75,
\end{align*}
which concludes the proof.
\end{proof}


\begin{proof}[Proof of Lemma \ref{NGBcosts}]
Let $e=(u,v)\in E$. Assume without loss of generality that $e$ is oriented to $u$ in the output of Algorithm \ref{NGBrounding}. We prove that:
\begin{align} \label{NGBcostInequality}
c_{e,u}\leq \frac{1}{1-\alpha}\cdot (x_{e,u}c_{e,u}+x_{e,v}c_{e,v}).
\end{align}
We consider two cases depending whether $e$ was oriented to $u$ in the \LS or the \GS of Algorithm \ref{NGBrounding}. Assume $e$ was oriented in the \LS, then $x_{e,u}>\alpha$. Thus,
\begin{align*}
c_{e,u}&\leq \frac{1}{\alpha}x_{e,u}c_{e,u}\\
&\leq \frac{1}{\alpha}(x_{e,u}c_{e,u}+x_{e,v}c_{e,v})\\
&\leq \frac{1}{1-\alpha}(x_{e,u}c_{e,u}+x_{e,v}c_{e,v}).
\end{align*}
The last inequality follows from the fact that $\alpha\geq 0.5$.

Assume $e$ was oriented to $u$ in the \GS. From Observation \ref{EdgeObserv} (which still holds in this case) $x_{e,u}\geq 1-\alpha$. Hence we deduce:
\begin{align*}
c_{e,u}&\leq \frac{1}{1-\alpha}x_{e,u}c_{e,u}\\
&\leq \frac{1}{1-\alpha}(x_{e,u}c_{e,u}+x_{e,v}c_{e,v}).
\end{align*}
Hence, we conclude that Inequality \ref{NGBcostInequality} holds. Summing Inequality \ref{NGBcostInequality} over all edges we obtain that the orientation cost of the output of Algorithm \ref{NGBrounding} is at most:
\begin{align*}
\sum_{e\in E}c_{e,u} &\leq \sum_{e\in E} \frac{1}{1-\alpha}(x_{e,u}c_{e,u}+x_{e,v}c_{e,v})\\
&=\frac{1}{1-\alpha}\cdot OPT_{LP}.
\end{align*}
The proof is concluded since $OPT_{LP}\leq C(T)$.
\end{proof}

%% file: IntegralityGapCost.tex
We show that using \LPstong for every $k\in \mathbb{N}$, one must loose in the total orientation cost when obtaining an approximation for the makespan that is strictly better than $2$. This is in contrast to the classic result of \cite{ShmoysT93} for which one can achieve an approximation factor of $2$ with respect to the makespan with no loss in the assignment cost. 
This result is formulated in Theorem \ref{Theorem:LowerBound}. 

\begin{proof}[Proof of Theorem \ref{Theorem:LowerBound}]
	Let $\epsilon > 0$ arbitrarily small and $k\in \mathbb{N}$. We show that for every $0\leq \gamma< 0.25$, there is an instance that is feasible to \LPstong, and any integral solution with makespan at most $1.75+\gamma$ has a cost of at least $\nicefrac{1}{(\gamma+0.75+2\epsilon)}\cdot OPT_{\LPstong}$.
	
	The instance we consider is shown in Figure \ref{IntegralityCostInstance}. The weight of the edge $e=(u,v)$ is $1-\epsilon$. Furthermore, the load values (self loops weights) are denoted by $q_u,q_v$. Note that a load value $q$ (whether it is $q_u$ or $q_v$) is split into $\nicefrac{(qk)}{\epsilon}$ self loops each with a weight of $\nicefrac{\epsilon}{k}$. Additionally, we set the assignment costs to be $c_{e,v}=1$ and $c_{e,u}=0$.
	
	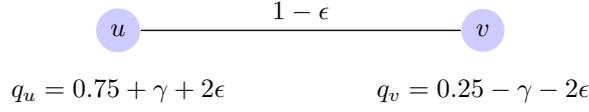
\begin{figure}
		\centering
		\begin{tikzpicture}
		[scale=.8,auto=left]
		\node[circle,fill=blue!20] (u) at (7,5)  {$u$};
		\node[circle,fill=blue!20] (v) at (13,5)  {$v$};
		
		\node[circle] (qu) at (7,4) {$q_u=0.75+\gamma+2\epsilon$};
		\node[circle] (qv) at (13,4) {$q_v=0.25-\gamma-2\epsilon$};
		
		\draw[draw=black, line width=0.5] (u) edge node{$1-\epsilon$} (v);
		
		\end{tikzpicture}
		\caption{Integrality gap instance for orientation cost for \LPstong.}
		\label{IntegralityCostInstance}
	\end{figure}
	
	There is exactly one feasible solution to \LPstong, which is $x_{e,u}=0.25-\gamma-2\epsilon$ and $x_{e,v}=0.75+\gamma+2\epsilon$. The cost of this fractional solution is $OPT_{\LPstong}=0.75+\gamma+2\epsilon$. There are two integral solutions:
	\begin{enumerate}
		\item  Orienting $e$ toward $u$. This solution has makespan of $1.75+\gamma+\epsilon$.
		
		\item Orienting $e$ toward $v$. This solution has an orientation cost of $1=\nicefrac{1}{(0.75+\gamma+2\epsilon)}\cdot OPT_{\LPstong}$.

	\end{enumerate}
	Therefore, every integral solution with makespan at most $1.75+\gamma$ has an orientation cost of at least $\nicefrac{1}{(0.75+\gamma+2\epsilon)}\cdot OPT_{\LPstong}$. Choosing $\epsilon'=\nicefrac{\epsilon}{2}$, and replacing $\epsilon$ with $\epsilon'$ in this proof, concludes the proof.
\end{proof}

%% file: GBUnrelatedLightHyper.tex
Let us recall the definition of \sGBUH{$\beta$}, where $\beta\in [0,1]$.
The input consists of a hypergraph, where each vertex represents a machine and each hyperedge represents a job.
The jobs are of two types, ``light'' and ``heavy''.
Every light hyperedge $e\in E$ is associated with weights $p_{e,u}$, one for each vertex $u\in e$ ({\em i.e.}, $e$ is unrelated since it has a different processing time for each of the machines it can be assigned to).
The requirement is that $p_{e,u}\leq \beta$ for evert $u\in e$.
On the other hand, every heavy hyperedge $e\in E$ must in fact be an edge, {\em i.e.}, $|e|=2$.
Such a heavy $e$ is associated with a single weight $p_e\in [0,1]$ ({\em i.e.}, $e$ is related since it has the same processing time for each of the two machines it can be assigned to).
In the above, as previously mentioned, we assume without loss of generality that the largest weight equals $1$.
For both types, light and heavy, orienting $e$ toward one of its endpoints is equivalent to assigning the job $e$ represents to the machine that is represented by the vertex $e$ was oriented to.
It is important to note that when $\beta=1$ the problem is exactly \sGAP, and when $\beta=0$ the problem is exactly \sGB.

Our result for \sGBUH{$\beta$} is summarized in Theorem \ref{Theorem:GBUH}, which improves upon the previous result of \cite{HuangO16} (refer to Section \ref{sec:Introduction} for a thorough discussion on how our result improves upon \cite{HuangO16}).
To the best of our knowledge, our result provides the first approximation better than $2$ with respect to the makespan of a natural problem that captures \sGB but is not captured by \sRA.

\begin{proof}
	Note that every instance of \sGBUH{$\beta$}, where $\beta<\nicefrac{1}{2}$, is an instance of \sGBUH{0.5}. Therefore it sufficient to prove the theorem for $\beta\geq \nicefrac{1}{2}$. Moreover, we say $e\in E$ is a \textit{heavy edge} if it is not a light hyperedge and $p_e>\beta$. Note that each edge is either a light hyperedge or a heavy edge.
	
	The proof is very similar to the one of Theorem \ref{GBBasicTheorem}. We use Algorithm \ref{GBrounding} with $f_\alpha$ (see \ref{ThresholdFunc}) as a threshold function (the same threshold function used in the proof of Theorem \ref{GBBasicTheorem}). The bounding of the orientation cost remains the same as in Lemma \ref{CostApprox} (which still holds in this case). Therefore, the orientation cost is at most $\nicefrac{1}{\alpha}\cdot C(T)$.
	
	However, the makespan bounding in the proof of Theorem \ref{GBBasicTheorem} does not hold. The main reason that the analysis is not the same is due to Observation \ref{GraphObserv}. This observation holds only for edges and not for hyperedges. However, since the hyperedges have a bounded processing time, we can achieve a similar result.
	In the makespan bounding in the proof of Theorem \ref{GBBasicTheorem}, we consider two cases.
	For simplicity of presentation, we repeat the analysis as in the proof of Theorem \ref{GBBasicTheorem} and elaborate on the differences.
	
	\noindent	\textbf{Case 1:} No edge was oriented toward $u$ in the \LS of Algorithm \ref{GBrounding}. Using the same notations in the proof of Theorem \ref{GBBasicTheorem}, if $e_1$ (the edge on top of $slot(u,1)$) is heavy, then the analysis remains valid. This is since $e_1$ is an edge and thus Observation \ref{GraphObserv} holds.
	Otherwise, $e_1$ is a light hyperedge, and thus $p_1\leq \beta$ 
	According to the original analysis of \cite{ShmoysT93} (or Lemma \ref{lemmaST}) the makespan is at most $1+\beta\leq 1+3\gamma+0.25\leq 1.75+\gamma$ (this follows from the fact that $\nicefrac[]{\beta}{3} - \nicefrac[]{1}{12}\leq \gamma \leq \nicefrac[]{1}{4}$).
	
	\noindent	\textbf{Case 2:} In the second case, an edge was oriented to $u$ in the \LS. In constrast to the proof of Theorem \ref{GBBasicTheorem}, it is no longer necessarily true to determine that $p_1\leq \nicefrac{1}{2}$. This is because we cannot apply Observation \ref{GraphObserv} to the edge $e_1$ if it is a light hyperedge (if it is an edge the proof remains valid). If $e_1$ is a light hyperedge, then the best bound on $p_1$ is $p_1\leq \beta$ (rather than $\nicefrac{1}{2}$). Using the bound on the load on $u$ from the proof of Theorem \ref{GBBasicTheorem}, we bound the load on $u$ by:
	\begin{align*}
	1+(1-\alpha)q_1+f_\alpha(p_1)p_1+(1-f_\alpha(p_1))p_2 &\leq 1+(1-\alpha)+p_1\\
	&\leq 2-\alpha+\beta\\
	&=1.5+\beta-2\gamma\\
	&\leq 1.75+\gamma.
	\end{align*}
	The equality follows from our choice of $\alpha$ to be $\alpha=2\gamma+0.5$. Moreover, the last inequality follows from the fact that $\nicefrac{\beta}{3}-\nicefrac{1}{12}\leq \gamma$.
	This concludes bounding the makespan, and thus concludes the proof.
\end{proof}

%% file: GBUnrelatedLightHyperUnrelatedHeavyEdge.tex
The problem of \sGBU{$\beta$} further generalizes the above \sGBUH{$\beta$} as it allows heavy edges to have unrelated weights.
Formally, every heavy edge $e=(u,v)\in E$ is associated with two weights $p_{e,u}$ and $p_{e,v}$, {\em i.e.}, $e$ is unrelated since $p_{e,u}$ indicates the processing time of the job $e$ represents on the machine that is represented by $u$.
The requirement is that $ p_{e,u},p_{e,v}\in (\beta,1]$.

First, we prove that without the assumption that the optimal makespan is at least $1$, the problem is as hard as \sGAP. Formally, the following lemma gives an approximation preserving reduction from \sGAP to \sGBU{$\beta$} (the single criterion version of these problems). 

\begin{lemma} \label{HardnessLemma}
	For every $0<\beta\leq 1$, if there is a $c$-approximation to \sGBU{$\beta$} (without the assumption that the optimal makespan is at least $1$), then there is a $c$-approximation to the \GAP. \footnote{We remark that by using a similar reduction, it is possible to show approximation equivalence for the bicriteria versions. Formally, if there is an $(a,b)$-approximation to \sGBU{$\beta$}, there is an $(a,b)$-approximation to \sGAP. We defer the proof for a later version of the paper.}
\end{lemma}

\begin{proof}
	For simplicity of presentation, we use the notations of jobs and machines instead of edges and vertices. Let ALG be a c-approximation to \sGBU{$\beta$}.
	Given an instance of \sGAP, denote the different (sorted) values of the processing times $\textbf{p}$ by $\infty = w_0 > w_1 > w_2 > \dots > w_k$. We present in Algorithm \ref{Reduction} our approximation preserving reduction.
	As previously mentioned in the paper, without loss of generality, we scale all processing times such that $max\{p_{i,j}\mid p_{i,j}\neq \infty\}=1$.
	
	\begin{algorithm*}[H] \label{Reduction}
		\caption{Reduction($\textbf{p}$)}
		\For{$\ell=1$ upto $k$}
		{If $p_{i,j} > w_\ell$, set $p_{i,j} \leftarrow \infty$.\\
			Add a new job $j'$ and two new machines $i_1$ and $i_2$ such that $p_{i_1,j'} = w_{\ell} + \epsilon$, $p_{i_2,j'} = \frac{1}{\beta}w_{\ell}$ and for every $i\notin \{i_1,i_2\}$ set $p_{i,j'} = \infty$. \\ 
			Run ALG on the new instance. 
			Reset all changes made to $\textbf{p}$ in this iteration.
		}
		Return the best schedule among all iterations (without the job $j'$ and machines $i_1,i_2$).
	\end{algorithm*}
	
	We now show that Algorithm \ref{Reduction} is a c-approximation to the given \sGAP instance.
	Denote by $OPT$ the optimal makespan of the original \sGAP instance.
	We note that there exists some $\ell$ such that $w_\ell \leq OPT < w_{\ell-1}$.
	We prove that for this $\ell$, the $\ell$th iteration of Algorithm \ref{Reduction} produces a schedule with makespan at most $c\cdot OPT + \epsilon$ where $\epsilon > 0 $ is arbitrarily small.
	
	First, note that the input is valid to ALG. This is because the only heavy job is $j'$ (the new job added in step 3), and the rest of the input (the jobs of the \sGAP instance), are light hyperedges. Thus, step 4 of Algorithm \ref{Reduction} is valid.
	Second, denote by $OPT'$ the optimal makespan of the instance created in iteration $\ell$. We show that $OPT'\leq OPT + \epsilon$, which concludes the proof. A feasible orientation can orient the new added job $j'$ toward $i_1$, and the rest of the input can be oriented like in the optimal solution to the GAP instance (since $OPT<w_{\ell-1}$). The makespan of this orientation is at most $OPT'$, and thus we conclude:
	\[
	OPT' \leq max\{OPT, w_i+\epsilon\} \leq max\{OPT, OPT+\epsilon\} = OPT+\epsilon.
	\]	
\end{proof}


Surprisingly, by adding a mild constraint on the value of the optimal makespan we can give an approximation better than $2$ for \sGBU{$\beta$}. The assumption is that $OPT\geq 1$, and we assume it from this point onwards. Notice that $OPT \geq 1$ holds naturally in \sGB, and \sGBUH{$\beta$}. Therefore, \sGBU{$\beta$} with the assumption that the value of the optimal makespan is at least $1$, captures these problems.

\subsubsection{Upper Bound}
We extend our previous results to \sGBU{$\beta$} and achieve the same approximation factors for $\beta\geq \sqrt{2}-1$. To be precise, in the case that $\beta\geq \sqrt{2}-1$, we get the same result as in Theorem \ref{Theorem:GBUH}. This result is summarized in Theorem \ref{Theorem:GBU}.\footnote{We remark that makespan approximation of less than 2 is still possible with this technique for $\beta<\sqrt{2}-1$. The required threshold function can be $f(p)=1-\nicefrac{1}{x}$ for $p>\nicefrac{1}{x}$ and 1 otherwise, for sufficiently big $x$ (depends on $\beta$).}

We use Algorithm \ref{GBrounding} to prove Theorem \ref{Theorem:GBU}. Recall that given a target makespan $T$ the algorithm scales all the processing times by $T$. Similarly to \sGB, if $T<1$, the algorithm outputs that there is no orientation with makespan at most $T$. This is true since $OPT\geq 1$. Hence, similarly to \sGB we know that $p_{e,u}\leq 1$ for every $e\in E$ and $u\in e$.

The only difference in the algorithm with respect to \sGB is that if $p_{e,u}\leq 0.75+\gamma$ for every $e\in E$ and $u\in e$, we use a rounding function $f\equiv 1$. This will result a $(1.75+\gamma, 1)$-approximation as proved in \cite{ShmoysT93} (or Lemma \ref{lemmaST}). Therefore, we assume the processing times are scaled by a multiplicative factor of at least $0.75+\gamma$, \ie the target makespan $T$ satisfies $T\geq \nicefrac{1}{(0.75+\gamma)}$.

We choose the following threshold function to prove Theorem \ref{Theorem:GBU}:
\begin{align} \label{ThresholdFuncGBULHUB}
f_\alpha(p_{i,j}) =
\begin{cases}
1 & \text{if } p_{i,j}\leq \nicefrac{1}{3} \\
\alpha & \text{if } p_{i,j}> \nicefrac{1}{3} \\
\end{cases}
\end{align}
where $max\{\nicefrac{1}{3}+\nicefrac{(2\beta)}{3}, \nicefrac{2}{3}\}\leq \alpha \leq 1$.

First we upper bound the assignment cost of Algorithm \ref{GBrounding} with the above threshold function \ref{ThresholdFuncGBULHUB}.
\begin{corollary} \label{CostCorollaryGBULHUB}
	For the specific choice of a threshold function $f_\alpha$ as in \ref{ThresholdFuncGBULHUB}, the output of Algorithm \ref{GBrounding} has an assignment cost of at most $\nicefrac{1}{\alpha}\cdot C(T)$.
\end{corollary}
\begin{proof}
	The proof follows immediately from Lemma \ref{CostApprox}, which still holds for \sGBU{$\beta$}.
\end{proof}

Now we focus on bounding the makespan.
We prove that the makespan is at most $1.5 + 0.5\alpha$, and choosing $\alpha=2\gamma+0.5$ gives the desired result.

Fix a machine $i$. We use the same notations and definition as in the proof of Theorem \ref{BetterGBTheorem}. Moreover, notice that observations \ref{SetObserv}, \ref{NumOriented}, and \ref{SumRounded}, are valid in this case as well. However, note that Observation \ref{GraphObserv} does not trivially hold, even for heavy edges. This is due to the fact that the heavy edges are unrelated. We now formulate the exact observation (similar to Observation \ref{GraphObserv}) that holds for \sGBU{$\beta$}.


\begin{observation} \label{AlmostGraphObserv}
	Let $j$ be a heavy job that isn't assigned in the \LS of Algorithm \ref{GBrounding}. Additionally, let $i$ be a machine such that $j$ could be assigned to $i$. Then, for our choice of a rounding function $f_\alpha$ it holds that $1-\alpha=1-f_\alpha(p_{i,j})\leq x_{i,j}\leq f_\alpha(p_{i,j})=\alpha$.
\end{observation}

\begin{proof}
	Let $j$ be a heavy job that isn't assigned in the \LS, and $i$ a machine such that $j$ could be assigned to $i$. The definition of a heavy job implies that before scaling the processing times (in the binary search for the optimal makespan): $p_{i,j}>\beta$.
	Moreover, since we scale by a factor of at least $0.75+\gamma$, we know that after scaling the processing times satisfy $p_{i,j}>\beta(0.75+\gamma)$. Since $\gamma\geq \nicefrac{\beta}{3} -\nicefrac{1}{12}$, we can deduce:
	\begin{align*}
	p_{i,j}&>\beta(0.75+\nicefrac{1}{3}\beta -\nicefrac{1}{12})\\
	&=\beta(\nicefrac{2}{3}+\nicefrac{1}{3}\beta)\\
	&\geq 0.5(\nicefrac{2}{3}+\nicefrac{1}{3}\cdot 0.5)\\
	&\geq \nicefrac{1}{3}.
	\end{align*}
	Therefore, if $x_{i,j}>\alpha$ according to Algorithm \ref{GBrounding} and the definition of $f_\alpha$ (see \ref{ThresholdFuncGBULHUB}), $j$ should be assigned to $i$ in the \LS, which is a contradiction. Moreover, if $x_{i,j}<1-\alpha$, then $x_{i',j}>\alpha$ where $i'$ is the only machine different from $i$, that $j$ could be assigned to. Since $p_{i',j}>\nicefrac{1}{3}$ as well, $j$ should be assigned to $i'$ in the \LS, again a contradiction.
\end{proof}

Now we continue with bounding the makespan. Fix a machine $i$. We show that the load on $i$ is at most $1.5+0.5\alpha$. We use the same definitions and notations as in the previous sections. Recall that $j_1,\dots, j_k$ are the jobs on top of the slots of machine $i$, and $j'_1,j'_2$ are the jobs assigned to machine $i$ in the \LS. For simplicity of notation we use $p_\ell$ and $q_\ell$ do denote $p_{i,j_\ell}$ and $p_{i,j'_\ell}$ respectively. We consider two cases, whether $j_1$ is a heavy job or not.

\noindent \textbf{Case 1:} In this case $j_1$ is a heavy job. Therefore, from Lemma \ref{lemmaST} and Observation \ref{AlmostGraphObserv} the load on $i$ is at most:
\begin{align*}
&1+(1-f_\alpha(q_1))q_1+(1-f_\alpha(q_2))q_2+f_\alpha(p_1)p_1+(1-f_\alpha(p_1))p_2\\
&=1+(1-\alpha)(q_1+q_2)+\alpha p_1+(1-\alpha)p_2\\
&=1+(1-\alpha)(q_1+p_1)+(1-\alpha)(q_2+p_1)+(3\alpha-2) p_1+(1-\alpha)p_{j_2}\\
&\leq 1+(1-\alpha)+(1-\alpha)+(3\alpha-2) +(1-\alpha)\cdot 0.5\\
&=1.5+0.5\alpha.
\end{align*}
The last inequality follows from Observations \ref{SumRounded} and \ref{SecondSlotObserv}, and the fact that $(3\alpha-2)p_1$ is maximized when $p_1=1$ (this is true since $\alpha\geq \nicefrac{2}{3}$).

\noindent \textbf{Case 2:} In this case $j_1$ is a light job. Therefore, $p_{j_1}\leq \beta$.\footnote{Since we assume that the optimal makespan is at least 1, when scaling the processing times they do not increase.} From Lemma \ref{lemmaST}, the load on $i$ is at most:
\begin{align*}
&1+(1-f_\alpha(q_1))q_1+(1-f_\alpha(q_2))q_2+f_\alpha(p_1)p_1+(1-f_\alpha(p_1))p_2\\
&\leq 1+(1-\alpha)(q_1+q_2)+p_1\\
&\leq 1+(1-\alpha)+\beta\\
&=2-\alpha+\beta\\
&\leq 1.5+0.5\alpha.
\end{align*}
The first inequality follows from the fact that $p_2\leq p_1$. Additionally, the second inequality follows from Observation \ref{SumRounded}. Lastly, the third inequality follows from the fact that $\alpha\geq\nicefrac{1}{3}+\nicefrac{2}{3}\beta$.

From the above we can deduce the following corollary:
\begin{corollary} \label{GBUHLUBMakespanCorrol}
	The makespan of the assignment produced by Algorithm \ref{GBrounding} with a threshold function $f_\alpha$ as in \ref{ThresholdFuncGBULHUB}) is at most $1.5+0.5\alpha$.
\end{corollary}

Now we are ready to prove Theorem \ref{Theorem:GBU}.
\begin{proof}[Proof of Theorem \ref{Theorem:GBU}]
	The proof follows from Corollaries \ref{CostCorollaryGBULHUB} and \ref{GBUHLUBMakespanCorrol}, and choosing $\alpha=2\gamma+0.5$.
\end{proof}

\subsubsection{Integrality Gaps}
\input{UnrelatedHeavyIntegrality.tex}

%% file: UnrelatedHeavyIntegrality.tex
We now present a matching lower bound for Theorem \ref{Theorem:GBU}. 
Specifically, we prove Theorem \ref{Theorem:GBUintegrality} by providing two integrality gaps to \LPstong.
Intuitively, the first instance shows that using \LPstong it is impossible to achieve an approximation less than $\nicefrac{11}{6}$ with respect to the makespan.
Intuitively, the second instance shows that using \LPstong an approximation of $1.75+\gamma$ with respect to the makespan must incur a loss of at least $\nicefrac{1}{(2\gamma +0.5)}\cdot OPT_{\LPstong}$ with respect to the orientation cost. This is formulated in the following two lemmas.
\begin{lemma} \label{GBUHMakespanIntegrality}
There exists an instance of \sGBUH{0.5} that is feasible to \LPstong, such that every integral solution has makespan of at least $\nicefrac{11}{6}-\epsilon$, for every $\epsilon>0$.
\end{lemma}

\begin{lemma} \label{GBUHCostIntegrality}
There exists an instance of \sGBUH{0.5} that is feasible to \LPstong, such that every integral solution with makespan of at most $1.75+\gamma$ has an assignment cost of at least $\nicefrac{(1-\epsilon)}{(2\gamma+0.5)}\cdot OPT_{\LPstong}$, for every $\epsilon>0$ and $\nicefrac{1}{12}\leq \gamma< \nicefrac{1}{4}$.
\end{lemma}

The proof of Theorem \ref{Theorem:GBUintegrality} follows immediately from Lemmas \ref{GBUHMakespanIntegrality} and \ref{GBUHCostIntegrality}. We now proceed to the proofs of these lemmas.

\begin{proof}[Proof of Lemma \ref{GBUHMakespanIntegrality}]
	Let $\epsilon>0$. The instance consist of $\nicefrac{1}{\epsilon}$ identical copies of a path of $\nicefrac{1}{\epsilon}$ vertices. We describe the $i$th path, as it is identical to the others. We denote by $u_i$ the leftmost vertex in the $i$th path. All the edges of the path are heavy. The edge incident to $u_i$ has a weight of $1-6\epsilon$, and the rest of the edges have a processing time of $0.5+\epsilon$ to the left endpoint and a processing time of $1-6\epsilon$ to the right endpoint. Additionally, each vertex has a load value (self loop) of $\nicefrac{1}{3}$. In fact, the load value is split into several self loops in the input, and the exact number will be determined later in the proof.

	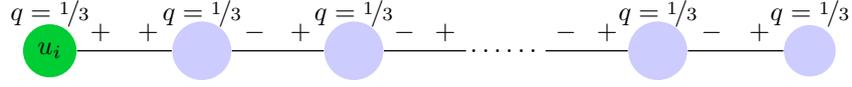
\begin{figure}
		\centering
		\begin{tikzpicture}
		[scale=.5,auto=left]
		\node[circle,fill=blue!20!green=20!] (u) at (0,5)  {$u_i$};
		\node[circle,fill=blue!20] (w1) at (4,5) {\color{blue!20}$w_1$};
		\node[circle,fill=blue!20] (w2) at (8,5) {\color{blue!20}$w_2$};
		\node[circle,fill=blue!20] (w3) at (16,5) {\color{blue!20}$w_3$};
		\node[circle,fill=blue!20] (v) at (20,5)  {\color{blue!20}$v_i$};
		
		\node[circle] (s1) at (11.3,5)  {};
		\node[circle] (s2) at (12.7,5)  {};
		
		\node[circle] (dots) at (12,5)  {$\dots\dots$};	
		
		\node[circle] (qu) at (0,6) {$q=\nicefrac{1}{3}$};
		\node[circle] (qw1) at (4,6) {$q=\nicefrac{1}{3}$};
		\node[circle] (qw2) at (8,6) {$q=\nicefrac{1}{3}$};
		\node[circle] (qw3) at (16,6) {$q=\nicefrac{1}{3}$};
		\node[circle] (qv) at (20,6) {$q=\nicefrac{1}{3}$};
		
		\draw[draw=black, line width=0.5] (u) edge node[near start]{$+$} node[near end]{$+$} (w1);
		\draw[draw=black, line width=0.5] (w1) edge node[near start]{$-$} node[near end]{$+$} (w2);
		\draw[draw=black, line width=0.5] (w3) edge node[near start]{$-$} node[near end]{$+$} (v);
		
		\draw[draw=black, line width=0.5] (w2) edge node[near start]{$-$} node[near end]{$+$} (s1);
		\draw[draw=black, line width=0.5] (s2) edge node[near start]{$-$} node[near end]{$+$} (w3);
		
		\end{tikzpicture}
		\caption{A single path from the makespan integrality gap instance for \sGBUH{0.5}. There is a dedicated load value of $\nicefrac{1}{3}$ on each vertex. All vertices excluding $u_i$ share a light hyperedge of weight $0.5$. Above the edges are the weight of the edge depending on the orientation, where the plus sign denotes $1-6\epsilon$, and the minus sign denotes $0.5+\epsilon$.}
		\label{IntegralityMakespanGBUH}
	\end{figure}

	All the vertices in the $i$th path, except $u_i$, share a light hyperedge of weight $0.5$, denoted by $e_i$. This concludes the description of a single path. In addition, all the vertices $u_i$ over all paths, share a light hyperedge of weight $0.5$, denoted by $e'$. A single path from the instance is shown in Figure \ref{IntegralityMakespanGBUH}.
	
	Now we show that this instance is feasible to \LPstong 
Consider the following fractional solution, denoted by \bx: each edge in each path is oriented $\nicefrac{2}{3}$ toward the left endpoint (and thus $\nicefrac{1}{3}$ toward the right endpoint), and each hyperedge is uniformly oriented between all the vertices it shares (\ie each hyperedge is oriented a fraction of $\epsilon$ to each of the $\nicefrac{1}{\epsilon}$ vertices it shares).
	Note that \bx is a feasible solution to \LPstong. First, the edge constraints trivially hold from the definition of \bx. Now we show the Load constraints hold for each vertex. Since the fractional solution is symmetric over all the paths, we consider a single path $i$. The fractional load on $u_i$ is:
	\[
	\frac{1}{3}+\frac{2}{3}\cdot (1-6\epsilon)+ \frac{1}{\epsilon}\cdot \frac{1}{2} = 1-3.5\epsilon.
	\]
	Additionally, the fractional load on the rest of the vertices in the path is at most:
	\[
	\frac{1}{3}+\frac{2}{3}\cdot (0.5+\epsilon)+\frac{1}{3}\cdot (1-6\epsilon) + \frac{1}{\epsilon}\cdot \frac{1}{2} = 1-\frac{5\epsilon}{6}.
	\]
	Thus, the Load constraints hold for all vertices.
The Star constraints also trivially hold from the definition of \bx. Lastly we show the Set constraints hold. We split the load value to $\nicefrac{k}{\epsilon}$ self loops each with weight $\nicefrac{\epsilon}{(3k)}$. Notice that this is sufficient so that Set constraints with respect to sets of size at most $k$ hold. This is true since for every vertex $v$ and $S\subseteq \delta(v)$, such that $\sum_{e\in S} p_{e,v}>1$, $S$ must contain at least two edges that are not self loops. Hence, from the definition of \bx we can deduce that $\sum_{e\in S} x_{e,v}\leq |S|-1$.

	Finally, we show that every integral solution to this instance has makespan of at least $\nicefrac{11}{6}-6\epsilon$ (and since we can replace $\epsilon$ with $\nicefrac{\epsilon}{6}$ we conclude the proof). If there exists a vertex such that two path edges are oriented toward it, {\em i.e.}, a collision, then the load on that vertex is at least $\nicefrac{11}{6}-5\epsilon$. Otherwise in each path, all the edges are oriented either to the right or the left and no collision is formed.
	First, we examine the case that there exists a path, such that all of the edges are oriented to the right.
Without loss of generality assume this is the $i$th path.
In this case, if we denote by $v$ the vertex that $e_i$ is oriented to, the load of $v$ is $\nicefrac{1}{3}+1-6\epsilon+0.5=\nicefrac{11}{6}-6\epsilon$.
	Otherwise, all the edges of all the paths are oriented toward the left. In particular, for every path $i$ the vertex the edge incident to $u_i$ in the path is oriented toward $u_i$. Thus, there exists a vertex $u_i$ such that the hyperedge $e'$ is oriented toward it. Therefore, the load on this $u_i$ is $\nicefrac{1}{3}+1-6\epsilon+0.5=\nicefrac{11}{6}-6\epsilon$.
	
	Hence, we can conclude that $\textbf{x}$ is a feasible solution to \LPstong, and every integral solution has makespan of at least $\nicefrac{11}{6}-6\epsilon$. By replacing $\epsilon$ with  $\nicefrac{\epsilon}{6}$ in this proof, we finish the proof.
\end{proof}

\begin{proof}[Proof of Lemma \ref{GBUHCostIntegrality}]
	Let $\epsilon>0$ and $\nicefrac{1}{12}\leq \gamma< \nicefrac{1}{4}$. First, define $\epsilon'=\nicefrac{\epsilon}{4}$ and $\gamma'= \nicefrac{\gamma}{(1-4\epsilon')}$. We now present an instance of \sGBUH{0.5} that is feasible to \LPstong and every integral solution with makespan at most $1.75+\gamma$ has an orientation cost of at least $\nicefrac{(1-\epsilon)}{(2\gamma+0.5)}\cdot OPT_{\LPstong}$. The instance is a cycle of $\nicefrac{1}{\epsilon'}$ vertices that share a hyperedge of weight $0.5$. The edges in the cycle are heavy and have a processing time of $1-\epsilon'$ when oriented counter-clockwise and $0.5+\epsilon'$ when oriented clockwise. Additionally, each vertex has a load value of $\gamma'+0.25-(4\gamma'+0.5)\epsilon'$. Lastly, one edge in the cycle, denoted by $e'$, has orientation cost of $1$ when oriented clockwise. The rest of the orientation costs are $0$. The instance is shown in Figure \ref{IntegralityCostGBUH}.

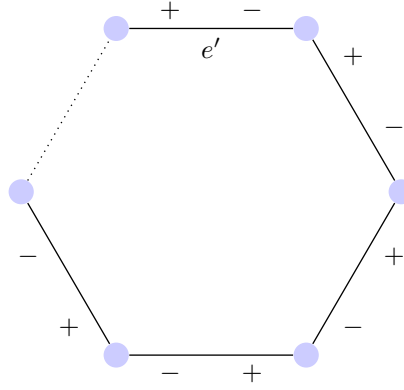
\begin{figure}
	\centering
	\begin{tikzpicture}
	[scale=2.5,auto=left]
	\node[circle,fill=blue!20] (u1) at (1,0)  {};
	\node[circle,fill=blue!20] (u2) at (0.5,0.866) {};
	\node[circle,fill=blue!20] (u3) at (-0.5,0.866) {};
	\node[circle,fill=blue!20] (u4) at (-1,0) {};
	\node[circle,fill=blue!20] (u5) at (-0.5,-0.866)  {};
	\node[circle,fill=blue!20] (u6) at (0.5,-0.866)  {};
	
	\draw[draw=black, line width=0.5] (u2) edge node[near start]{$+$} node[near end]{$-$} (u1);
	
	\draw[draw=black, line width=0.5] (u3) edge node[near start]{$+$} node[near end]{$-$} node[below]{$e'$} (u2);
	
	\draw[draw=black, line width=0.5, dotted] (u4) edge (u3);
	
	\draw[draw=black, line width=0.5] (u5) edge node[near start]{$+$} node[near end]{$-$} (u4);
	\draw[draw=black, line width=0.5] (u6) edge node[near start]{$+$} node[near end]{$-$} (u5);
	\draw[draw=black, line width=0.5] (u1) edge node[near start]{$+$} node[near end]{$-$} (u6);

	\end{tikzpicture}
	\caption{The instance of the integrality gap in the proof of Lemma \ref{GBUHCostIntegrality}. Each vertex has a load value of $\gamma'+0.25-(4\gamma'+0.5)\epsilon'$, and they share a single hyperedge of weight $0.5$. The plus and minus signs of the edges represent the weight with respect to each endpoint. The plus sign represents a weight of $1-\epsilon'$ and the minus sign represents a weight of $0.5+\epsilon'$. Orienting $e'$ clockwise, incurs a cost of $1$.}
	\label{IntegralityCostGBUH}
\end{figure}

Now we show this instance is feasible to \LPstong 
In fact there is exactly one feasible solution, which we denote by \bx. This fractional solution is described as follows: each edge in the cycle is oriented $2\gamma'+0.5$ clockwise (and thus $0.5-2\gamma'$ counter-clockwise). Moreover, the hyperedge is uniformly oriented to all the vertices in the cycle, \ie, it is fractionally oriented $\epsilon'$ to each of the $\nicefrac{1}{\epsilon'}$ vertices it shares.

Now, we show \bx is feasible to \LPstong. The Edge and Star constraints trivially hold from the definition of \bx. Since the solution \bx is symmetric to all the vertices, we show the Load and Set constraints are valid for an arbitrary vertex $u$. The fractional load on $u$ is:
\[
\gamma'+0.25-(4\gamma'+0.5)\epsilon' + (2\gamma'+0.5)(0.5+\epsilon') + (0.5-2\gamma')(1-\epsilon')+\nicefrac{1}{\epsilon'}\cdot 0.5 = 1.
\]
Thus, the Load constraints hold. If we split the load value to $\nicefrac{k}{\epsilon'}$ self loops, then similarly to the proof of Lemma \ref{GBUHMakespanIntegrality}, a set $S\subseteq \delta(u)$ such that $\sum_{e\in S} p_{e,u}>1$ and $|S|\leq k$, must contain at least two edges that are not self loops. Thus, from the definition of \bx the Set constraints hold, and therefore \bx is feasible.

Now we show each integral solution with makespan at most $1.75+\gamma'$ has orientation cost of at least $\nicefrac{(1-4\epsilon')}{(2\gamma'+0.5)}\cdot OPT_{\LPstong}$. First, observe that the orientation cost of \bx is $2\gamma'+0.5$, and since \bx is the only feasible solution to the relaxation then $OPT_{\LPstong}=2\gamma'+0.5$.
Any integral solution such that two cycle edges are oriented to the same vertex (collision) has makespan of at least:
\begin{align*}
(\gamma'+0.25-(4\gamma'+0.5)\epsilon')+1-\epsilon'+0.5+\epsilon'&=1.75+\gamma'-(4\gamma'+0.5)\epsilon'\\
&=1.75 + (1-4\epsilon')\gamma'+0.5\epsilon'\\
&=1.75+\gamma+\nicefrac{\epsilon}{8}.
\end{align*}
If there are no collisions, all the edges in the cycle are either oriented clockwise or counter-clockwise. If the edges are oriented clockwise, then in particular, $e'$ is oriented clockwise. Thus, the orientation cost of this solution is $1$. Otherwise, all the cycle edges are oriented counter-clockwise. Let $u$ be the vertex such that the light hyperedge is oriented to it. Therefore, the load on $u$, if all the edges are oriented counter-clockwise, is:
\begin{align*}
(\gamma'+0.25-(4\gamma'+0.5)\epsilon')+1-\epsilon'+0.5&=1.75+\gamma'-(4\gamma'+1.5)\epsilon'\\
&= 1.75+(1-4\epsilon')\gamma'+1.5\epsilon'\\
&=1.75+\gamma+\nicefrac{3\epsilon}{8}.
\end{align*}
Thus, every integral solution with makespan at most $1.75+\gamma$ has a cost of $1$. We conclude the proof since:
\begin{align*}
\frac{1-\epsilon}{2\gamma+0.5}\cdot OPT_{\LPstong} &= \frac{1-\epsilon}{2\gamma+0.5}\cdot (2\gamma'+0.5)\\
&=\frac{1-\epsilon}{2\gamma+0.5}\cdot \left(\frac{2\gamma}{1-4\epsilon'}+0.5\right)\\
&\leq \frac{1-\epsilon}{2\gamma+0.5}\cdot \frac{2\gamma+0.5}{1-4\epsilon'}\\
&=1.
\end{align*}
\end{proof}

\begin{proof}[Proof of Theorem \ref{Theorem:GBUintegrality}]
	Follows immediately from Lemmas \ref{GBUHMakespanIntegrality}, \ref{GBUHCostIntegrality}.
\end{proof}

%% file: SemiRestrictedGB.tex
Consider the general problem of \textsc{Unrelated Graph Balancing}, 
which is identical to \sGB except that an edge can have a different weight depending on its orientation: $p_{e,u}$ and $p_{e,v}$ for every $ e=(u,v)\in E$, {\em i.e.}, the weights are unrelated.
This generalization of \sGB was already considered in \cite{VerschaeW14, EbenlendrKS14}, who presented lower bounds for the problem.
Specifically, they showed that the even the configuration LP (which captures \LPstong) has an integrality gap of $2$ with respect to the makespan.

We consider an interesting special case of the above problem where the weights are still unrelated, but cannot vary arbitrarily. Formally, each edge $e=(u,v)\in E$ has two weights depending on the vertex $e$ is oriented to, which satisfy: $p_{e,u}\leq c\cdot p_{e,v}$ and $p_{e,v}\leq c\cdot p_{e,u}$ (where $c\geq 1$ is a parameter of the problem).
We denote this problem by \SRGB (\sSRGB).

Our result for \sSRGB is formulated in Theorem \ref{Theorem:SRGB}.
Note that \sSRGB captures \sGB when $c=1$, and indeed in Theorem \ref{Theorem:SRGB} we achieve a $(\nicefrac{11}{6}, \nicefrac{3}{2})$-approximation for \sSRGB when $c=1$ (similarly to Theorem \ref{GBBasicTheorem}).
Moreover, when $c=\infty$ Theorem \ref{Theorem:SRGB} achieves a $(2,1)$-approximation for \sSRGB, matching the integrality gap of \cite{VerschaeW14, EbenlendrKS14}.
Finally, we also show that in general Theorem \ref{Theorem:SRGB} provides a $(2-\Omega(\nicefrac[]{1}{c}),1+O(\nicefrac[]{1}{c}))$-approximation for \sSRGB.
Figure \ref{SemiRelatedApproxGraph} shows the makespan approximation obtained in Theorem \ref{Theorem:SRGB} as a function of $c$.

\begin{center}
	\begin{figure}
		\begin{center}
			\includegraphics[scale=0.55]{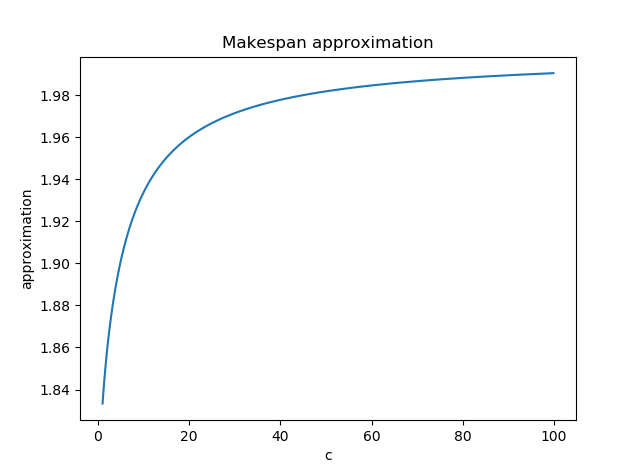}
			\caption{Makespan approximation as a function of the value $c$.}
			\label{SemiRelatedApproxGraph}
		\end{center}
	\end{figure}
\end{center}

In order to prove Theorem \ref{Theorem:SRGB} we use Algorithm \ref{GBrounding} and \LPstong (replacing $p_{e}$ with $p_{e,u}$) with a suitable choice of a threshold function $f$. We use the following threshold function:

\begin{align} \label{SemiThreshFunc}
f_{a,b}(p_e) =
\begin{cases}
a & \text{if } p_e>b \\
1 &					\text{if } p_e \leq b\\
\end{cases}
\end{align}
where the parameters $a$ and $b$ will be chosen shortly.

First, we present the following lemma that bounds the makespan of the orientation produced by Algorithm \ref{GBrounding} with threshold function $f_{a,b}$ (see \ref{SemiThreshFunc}):

\begin{lemma} \label{LoadLemma}
	The makespan of the orientation produced by Algorithm \ref{GBrounding} with $f_{a,b}$ is at most:
	\[
	max\left\{
	\frac{1}{a}+c\cdot b,\
	1.5+0.5a,\
	2-\left( 2-\frac{1}{a}\right)b
	\right\}
	\]
\end{lemma}

Using Lemma \ref{LoadLemma}, we can finish the proof of Theorem \ref{Theorem:SRGB}:
\begin{proof}[Proof of Theorem \ref{Theorem:SRGB}]
	First, we choose the values of $a$ and $b$ (as in the rounding function $f_{a,b}$) such that the terms in the maximum expression of Lemma \ref{LoadLemma} are balanced so as to minimize the makespan. 
	Thus, we choose $a$ and $b$ such that:
	\[
	\frac{1}{a}+c\cdot b=1.5+0.5a=2-\left(2-\frac{1}{a}\right)b.
	\]
	This is equivalent to finding an $a$ such that:
	\[
	\left(\frac{1}{c}+\frac{1}{2}\right)a^3+\left(\frac{5}{2c}-\frac{1}{2}\right)a^2 - \frac{7}{2c}a + \frac{1}{c}=0.
	\]
	
	We show there is a unique $a$ in the range $[0.5, 1-\nicefrac{1}{4c}]$, that satisfies the above equation. This fact, along with Lemma \ref{CostApprox} (which bounds the orientation cost of the algorithm) finishes the proof.
	
	Thus, we show that the polynomial:
	\[
	g(a) = \left(\frac{1}{c}+\frac{1}{2}\right)a^3+\left(\frac{5}{2c}-\frac{1}{2}\right)a^2 - \frac{7}{2c}a + \frac{1}{c}
	\]
	has a unique root in $[0.5,1-\nicefrac{1}{4c}]$. It can be proved that $g(0.5)=-\nicefrac{1}{16}$ and $g(1-\nicefrac{1}{4c})>0$, for every $c\geq 1$. Thus, from the Intermediate Value Theorem, we can deduce there is a root of the polynomial $g(a)$ in the range $[0.5,1-\nicefrac{1}{4c}]$. Furthermore, It can be shown using simple calculus that $g$ has a unique root in $[0.5,1]$ (by showing $g$ has no local maximum in this range). Hence, we can deduce that there is a unique root for $g(a)$ in the range $[0.5, 1-\nicefrac{1}{4c}]$.
\end{proof}

Now we prove Lemma \ref{LoadLemma} which bounds the makespan of the orientation produced by Algorithm \ref{GBrounding} with threshold function $f_{a,b}$.

\begin{proof}[Proof of Lemma \ref{LoadLemma}]
	Fix a vertex $u\in V$. We use definitions similar to the previous sections: denote by $e_1,\dots, e_k$ the edges on top of $slot(u,1),\dots, slot(u,k)$. Additionally, we denote by $p_1,\dots, p_k$ the weights (with respect to the vertex $u$) of those edges respectively.
	
	Now, we are required to introduce an observation similar to Observation \ref{GraphObserv}. The reason Observation \ref{GraphObserv} does not hold in this case, is due to the fact that an edge does not have the same weight on the two vertices it shares. However, since the weights cannot differ much, we observe the following fact.
	
	\begin{observation} \label{EdgeObserv}
		Let $e=(u,v)\in E$ such that $p_{e,u}>c\cdot b$ and $e$ is not oriented in the \LS of Algorithm \ref{GBrounding}. Then, $x_{e,u}\geq 1-a$.
	\end{observation}
	
	\begin{proof}
		From the problem definition, we know that $p_{e,u}\leq c\cdot p_{e,v}$. Since $p_{e,u}>c\cdot b$ we can deduce that $p_{e,v}>b$. In addition, since $e$ is not oriented in step 1 of Algorithm \ref{GBrounding}, then $x_{e,v}\leq a$. From the Edge constraint of $e$, we know that $x_{e,u}+x_{e,v}=1$. Therefore, we conclude that $x_{e,u}\geq 1-a$.
	\end{proof}

	Before proceeding to the proof of Lemma \ref{LoadLemma}, we denote by $E_{orient}$ the set of edges oriented toward $u$ in the \LS. In addition, we denote by $X$ the fractional load of edges oriented toward $u$ in the \LS:
	\[
	X\triangleq \sum_{e\in E_{orient}} x_{e,u}p_{e,u}.
	\]
	The next observation bounds the load from edges oriented toward $u$ in the \LS.
	
	\begin{observation} \label{SemiStep1}
		The load from edges oriented toward $u$ in the \LS of Algorithm \ref{GBrounding} is at most $\nicefrac{1}{a}\cdot X$.
	\end{observation}
	
	\begin{proof}
		The load on $u$ from edges oriented in the \LS is:
		\begin{align*}
		\sum_{e\in E_{orient}} p_{e,u} \leq \sum_{e\in E_{orient}} \frac{1}{a}\cdot x_{e,u}p_{e,u}=\frac{1}{a}X.
		\end{align*}
		The first inequality holds since for every $e\in E_{orient}$ it must be the case that $x_{e,u}>a$.
	\end{proof}
	
	We consider three cases, depending on the values of $p_1$ and $X$.
	
	\noindent	\textbf{Case 1:} $p_1\leq c\cdot b$. We can use the result by \cite{ShmoysT93} (or Lemma \ref{lemmaST}), in order to bound the load from edges oriented to $u$ in the \GS of Algorithm \ref{GBrounding} by:
	\[
	\sum_{i=1}^{k}p_{i}\leq 1-X + p_1.
	\]
	Note that this is true since the total fractional load on $u$ in the beginning of the \GS of Algorithm \ref{GBrounding} is $1-X$. From Observation \ref{SemiStep1} and the previous inequality we can bound the total load on $u$ as follows:
	
	\begin{align*}
	\frac{1}{a}X+(1-X)+p_1 &\leq 1 + \left(\frac{1}{a}-1\right)X+c\cdot b\\
	&\leq 1 +\left(\frac{1}{a}-1\right)+c\cdot b\\
	&= \frac{1}{a} +c\cdot b.
	\end{align*}
	
	\noindent	\textbf{Case 2:} $p_1>c\cdot b$ and $X=0$. Note that since $X=0$, no edge was oriented toward to $u$ in the \LS. From Observation \ref{EdgeObserv} we derive that $x_{e_1,u}\geq 1-a$. Therefore, using Lemma \ref{lemmaST} we obtain the following bound on the load on $u$:
	\[
	1+(1-a)p_1+a\cdot p_2.
	\]
	Since no edges were oriented to $u$ in the \LS, we can conclude that the load on $u$ is at most:
	\begin{align*}
	1+(1-a)p_1+a\cdot p_2&
	\leq 1+ 1-a + a\cdot 0.5\\
	&= 1.5 + 0.5a.
	\end{align*}
	Note that $p_2\leq \nicefrac{1}{2}$ from Observation \ref{SecondSlotObserv}.
	
	\noindent	\textbf{Case 3:} $p_1>c\cdot b$ and $X>0$. We start with an observation that follows from the Set constraints. The observation gives an upper bound to $p_1$.
	\begin{observation} \label{case3p1}
		$p_1\leq 1-b$.
	\end{observation}
	
	\begin{proof}
		Due to the case condition $X>0$, there is an edge $e'$ that was oriented toward $u$ in the \LS. From the definition of the rounding function \ref{SemiThreshFunc} it holds that $x_{e',u}>a$. Moreover, by Observation \ref{EdgeObserv} we obtain that $x_{e_1,u}\geq 1-a$. Thus, we can deduce: $x_{e_1,u}+x_{e',u}>1$.
		
		From the Set constraint of $S=\{e',e_1\}$ we derive that $p_1+p_{e',u}\leq 1$ (using Observation \ref{SetObserv}). Moreover, from the definition of the threshold function \ref{SemiThreshFunc} it holds that $p_{e',u}>b$. Therefore, we can conclude that $p_1\leq 1-p_{e',u}<1-b$.
	\end{proof}
	
	The next observation, gives an upper bound on $X$ that follows from the Load constraint on $u$.
	\begin{observation} \label{case3X}
		$X\leq 1-(1-a)\cdot p_1-(1-a)p_2$.
	\end{observation}
	\begin{proof}
		From the Load constraint on $u$ we deduce:
		\begin{align*}
		1\geq \sum_{e\in \delta(u)} x_{e,u}p_{e,u} &\geq X + \sum_{e\in slot(u,1)} x_{e,u}p_{e,u}\\
		&\geq X + (1-a)p_1 + a p_2.
		\end{align*}
		The last inequality follows from Observation \ref{EdgeObserv}, and the fact that the capacity of $slot(u,1)$ is exactly 1 (recall that without loss of generality we can add edges with weight zero in order to fill the slot if needed).
		From the previous inequality we conclude:
		\[
		X\leq 1-(1-a)p_1-a\cdot p_2.
		\]
	\end{proof}
	
	Now we bound the load on $u$. We use the same bound on the load of the edges oriented toward $u$ in the \GS, as in the previous case. In addition, we bound the load from edges oriented toward $u$ in the \LS by using Observation \ref{SemiStep1}. Thus, the load on $u$ is at most:
	\begin{align*}
	\frac{1}{a}X+(1-X)+a\cdot p_1+(1-a)p_2
	&= 1 + \left(\frac{1}{a}-1\right)X +a\cdot p_1+(1-a)p_2\\
	&\leq 1 + \left(\frac{1}{a}-1\right)(1-(1-a)p_1-a p_2) +a\cdot p_1+(1-a)p_2\\
	&= \frac{1}{a}+\left(2-\frac{1}{a}\right)p_1\\
	&\leq \frac{1}{a}+\left(2-\frac{1}{a}\right)(1-b)\\
	&= 2-\left(2-\frac{1}{a}\right)b.
	\end{align*}
	The first and second inequalities follow from Observations \ref{case3X} and \ref{case3p1} respectively.
	Note that in each case the load on $u$ is at most:
	$max\left\{
	\frac{1}{a}+c\cdot b,\
	1.5+0.5a,\
	2-(2-\frac{1}{a})b
	\right\}$, which concludes the proof.
\end{proof}